\documentclass[a4paper]{article}
\pdfoutput=1

\usepackage[utf8]{inputenc}
\usepackage[T1]{fontenc}
\usepackage{fullpage}
\usepackage{lmodern}

\usepackage{comment}
\usepackage{amsmath,amssymb,amsthm}
\usepackage{graphicx}
\usepackage{xcolor}
\usepackage{todonotes}
\usepackage{microtype}
\usepackage{xspace}
\usepackage[pdftex, plainpages = false, pdfpagelabels,
                 bookmarks=false,
                 bookmarksopen = true,
                 bookmarksnumbered = true,
                 breaklinks = true,
                 linktocpage,
                 pagebackref,
                 colorlinks = true,  
                 linkcolor = blue,
                 urlcolor  = blue,
                 citecolor = blue,
                 anchorcolor = blue,
                 hyperindex = true,
                 hyperfigures
                 ]{hyperref}
\usepackage[nameinlink]{cleveref}
\crefname{prop_ind_1}{}{}
\creflabelformat{prop_ind_1}{#2(i)#3}
\crefname{prop_ind_2}{}{}
\creflabelformat{prop_ind_2}{#2(ii)#3}
\usepackage{thm-restate}
\usepackage[ruled]{algorithm2e}

\usepackage{xargs}
\usepackage{tikz}
\usetikzlibrary{calc}
\usetikzlibrary{arrows,shapes}
\usetikzlibrary{decorations.pathmorphing}
\usetikzlibrary{positioning}

\usepackage{framed}
\usepackage{ctable}
\usepackage{xspace}
\usepackage{xifthen}
\newenvironment{tightcenter}
 {\parskip=0pt\par\nopagebreak\centering}
 {\par\noindent\ignorespacesafterend}

\newlength{\RoundedBoxWidth}
\newsavebox{\GrayRoundedBox}
\newenvironment{GrayBox}[1]%
   {\setlength{\RoundedBoxWidth}{\textwidth-4.5ex}
    \def\boxheading{#1}
    \begin{lrbox}{\GrayRoundedBox}
       \begin{minipage}{\RoundedBoxWidth}%
   }{%
       \end{minipage}
    \end{lrbox}%
    \begin{tightcenter}%
    \begin{tikzpicture}%
       \node(Text)[draw=black!20,fill=white,rounded corners,%
             inner sep=2ex,text width=\RoundedBoxWidth]%
             {\usebox{\GrayRoundedBox}};
        \coordinate(x) at (current bounding box.north west);
        \node [draw=white,rectangle,inner sep=3pt,anchor=north west,fill=white]
        at ($(x)+(6pt,.75em)$) {\boxheading};
    \end{tikzpicture}
    \end{tightcenter}\vspace{0pt}%
    \ignorespacesafterend
}

\newenvironment{XProblem}[2][]{\noindent\ignorespaces%
                                \FrameSep=6pt%
                                \parindent=0pt%
                \vspace*{-.5em}
                \ifthenelse{\isempty{#1}}{%
                  \begin{GrayBox}{\textsc{#2}}%
                }{%
                  \begin{GrayBox}{\textsc{#2} parametrised by~{#1}}%
                }
                \newcommand\Prob{Task:}%
                \newcommand\Input{Input:}%
                \begin{tabular*}{\textwidth}{@{\hspace{.1em}} >{\itshape} p{1.6cm} p{0.8\textwidth} @{}}%
            }{
                \end{tabular*}%
                \end{GrayBox}%
                \vspace*{-.5em}
                \ignorespacesafterend
            }

\newcommand{\wf}{\mathcal{C}}
\newcommand{\GF}{\text{GF}}

\newcommand{\thenet}{\textnormal{\textsf{Net}}}
\newcommand{\thenetg}{\textnormal{\textsf{Net graph}}}
\newcommand{\thenetgs}{\textnormal{\textsf{Net graphs}}}

\newtheorem{claim}{Claim}
\newtheorem{corollary}{Corollary}
\newtheorem{lemma}{Lemma}

\newtheorem{proposition}{Proposition}

\newtheorem{observation}{Observation}

\newcommand{\name}[1]{\textnormal{\textsc{#1}}}

\newcommand{\pname}{\textsc}
\newcommand{\ProblemFormat}[1]{\pname{#1}}
\newcommand{\ProblemIndex}[1]{\index{problem!\ProblemFormat{#1}}}
\newcommand{\ProblemName}[1]{\ProblemFormat{#1}\ProblemIndex{#1}{}\xspace}

\newcommand{\probSP}{\ProblemName{Xor Constrained Shortest Path}}
\newcommand{\probSC}{\ProblemName{Xor Constrained Shortest Cycle}}
\newcommand{\probCutUncut}{\ProblemName{Two-Sets Cut-Uncut}}
\newcommand{\probND}{\ProblemName{Network Diversion}}

\newcommand{\probGND}{\ProblemName{Generalized Network Diversion}}

\newcommand{\probTwoDCS}{\ProblemName{2-Disjoint Connected Subgraphs}}

\DeclareMathOperator{\cut}{cut}

\DeclareMathOperator{\operatorClassP}{P}
\newcommand{\classP}{\ensuremath{\operatorClassP}}
\DeclareMathOperator{\operatorClassNP}{NP}
\newcommand{\classNP}{\ensuremath{\operatorClassNP}}

\DeclareMathOperator{\operatorClassFPT}{FPT\xspace}
\newcommand{\classFPT}{\ensuremath{\operatorClassFPT}\xspace}
\DeclareMathOperator{\operatorClassW}{W}
\newcommand{\classW}[1]{\ensuremath{\operatorClassW[#1]}}

\newcommand{\Oh}{\mathcal{O}}

\title{Two-sets cut-uncut on planar graphs\thanks{The research leading to these results has received funding from the Research Council of Norway via the project BWCA (grant no. 314528). Matthias Bentert is supported by the European Research Council (ERC) under the European Union’s Horizon 2020 research and innovation programme (grant agreement No. 819416).}}
\author{Matthias Bentert\thanks{
Department of Informatics, University of Bergen, Norway. Emails:
\texttt{matthias.bentert@uib.no},
\texttt{pal.drange@uib.no},
\texttt{fomin@ii.uib.no},
\texttt{petr.golovach@uib.no},
\texttt{tuukka.korhonen@uib.no}
}
\and
Pål Grønås Drange\addtocounter{footnote}{-1}\footnotemark{}
\and
Fedor V.\ Fomin\addtocounter{footnote}{-1}\footnotemark{}
\and Petr A.\ Golovach\addtocounter{footnote}{-1}\footnotemark{}
\and Tuukka Korhonen\addtocounter{footnote}{-1}\footnotemark{}}
\date{}

\begin{document}
\maketitle

\begin{abstract}
  We study the following \probCutUncut problem on planar graphs.
  Therein, one is given an undirected planar graph $G$ and two sets of vertices~$S$ and~$T$.
  The question is, what is the minimum number of edges to remove from $G$,
  such that we separate all of $S$ from all of $T$, while maintaining that every vertex in $S$, and respectively in $T$, stays in the same connected component.
  %
  We show that this problem can be solved in time~$2^{|S|+|T|} n^{\Oh{(1)}}$ with a one-sided error randomized algorithm.
  Our algorithm implies a polynomial-time algorithm for the network diversion problem on planar graphs, which resolves an open question from the literature.
  More generally,  we show that  \probCutUncut remains fixed-parameter tractable even when parameterized by the number $r$ of faces 
  in the plane graph covering  the terminals $S \cup T$, by providing an algorithm of running time  $4^{r + \Oh(\sqrt r)} n^{\Oh(1)}$.
  
  
\end{abstract}

\section{Introduction}
\label{sec:introduction}

We consider the following variant of the \emph{cut-uncut problem}.
A \emph{cut} in a graph $G=(V,E)$ is a partitioning~$(A,B)$ of $V$, and we denote by $\cut_G(A)$ the \emph{cut-set}, that is, the set of edges with one endpoint in $A$ and the other in $B=V\setminus A$.
For two disjoint sets of vertices $S$ and $T$, $(A,B)$ is an \emph{$S$-$T$-cut} if $S\subseteq A$ and~$T\subseteq B$.

\begin{XProblem}{\probCutUncut}
  \Input & A graph $G$, two disjoint terminal sets $S,T\subseteq V(G)$, and an integer $k\geq 0$. \\
  \Prob &  Decide whether there exists an {$S$-$T$-cut} $(A,B)$ of $G$ with
$\left| \cut_G(A) \right| \leq k$
such that the vertices of  $S$ are in the same connected component of $G[A]$ and the vertices of $T$ are in the same connected component of $G[B]$.
\end{XProblem}

Our interest in \probCutUncut is two-fold. First, \probCutUncut is a natural optimization variant of the \probTwoDCS{}
problem that received considerable attention from the graph-algorithms and computational-geometry communities \cite{CyganPPW14a,GrayKLS12,KammerT12,PaulusmaR11,DBLP:conf/wg/TelleV13,HofPW09}. In this problem, one asks whether, for two given disjoint sets~$S,T\subseteq V(G)$, one can find disjoint sets~$A_1\supseteq S$ and~$A_2\supseteq T$ such that the subgraphs of $G$ induced by $A_i$, $i=1,2$, are connected.
In   \probCutUncut we not only want to decide whether there are disjoint connected sets containing terminal sets~$S$ and~$T$, but also minimize the size of the corresponding cut (if it exists).  Van 't Hof et al.~\cite{HofPW09} showed that \probTwoDCS{} is \classNP-complete in general graphs, even if $|S| = 2$, and Gray et al.~\cite{GrayKLS12} proved that the problem is \classNP-complete on planar graphs. This implies that  \probCutUncut is also \classNP-complete on planar graphs. 

%
%

Second,  \probCutUncut is closely related to the \probND problem, which has been studied extensively by the operations research and networks communities~\cite{cintron-arias2001networkdiversion,cullenbine2013theoreticalcomputational,curet2001networkdiversion,duan2014connectivitypreserving,erkenozgur2002branchandboundalgorithm,kallemyn2015modelingnetwork,lee2019combinatorial}.
In this problem, we are given an undirected graph $G$, two terminal vertices $s$ and $t$,  an edge $b=uv$, and an integer $k$. The task is
to decide whether it is possible to delete at most $k$ edges such that the edge $b$ will become a bridge with $s$ on one side and $t$ on the other.  Equivalently, the task is to decide whether there exists a \emph{minimal} $s$-$t$-cut of size at most~$k+1$ containing $b$.
While this problem seems very similar to the classic $s$-$t$-\ProblemName{Minimum Cut} problem, the complexity status of this problem (\classP{} vs.\ \classNP) is widely open.
Let us observe that a polynomial-time algorithm for the special case of   \probCutUncut with~$|S|=|T|=2$ implies a polynomial time algorithm for  \probND:
There are two cases, either~$s$ is in the same component as $u$, or $s$ is in the same component as $v$, and these correspond to instances of \probCutUncut with $S=\{s,u\}$ and $T=\{t,v\}$, and respectively, $S=\{s,v\}$ and~$T=\{t,u\}$.

\probND has important applications in transportation networks and has
therefore also been studied on planar graphs.
Cullenbine et al.~\cite{cullenbine2013theoreticalcomputational} gave a polynomial time algorithm for \probND on planar graphs for the special case when both terminals $s$ and $t$ are located on the same face.
They posed as an open problem whether this polynomial-time algorithm can be generalized to work on arbitrary planar graphs~\cite{cullenbine2013theoreticalcomputational}.
Duan et al.\ put out a preprint~\cite{duan2015ddosattack}, which among
other results, claims an algorithm resolving \probND on planar graphs in
polynomial time, but without a description of the algorithm.
We were not able to verify the correctness of the result due to
several missing details.
%
%
%
The result, however, is an immediate consequence of our main contribution, \Cref{thm:face-cover-param},  establishing the fixed-parameter tractability of  \probCutUncut{} on planar graphs parameterized by $|S|+|T|$.  \Cref{thm:face-cover-param}  also establishes a  more general result about fixed-parameter tractability of the problem parameterized by the minimum number of faces $r$ of the graph containing all terminals. (Notice that $r$ never exceeds $|S|+|T|$.) 


%
%

\begin{restatable}{theorem}{maintheoremfacecover}\label{thm:face-cover-param}
  There is a one-sided error randomized algorithm solving
  \probCutUncut{} on planar graphs in time $2^{|S|+|T|} \cdot n^{\Oh(1)}$.
  %
  %
  Moreover, there is a one-sided error randomized algorithm solving the
  problem in time $4^{r+\Oh(\sqrt{r})}\cdot n^{\Oh(1)}$, where $r$ is
  the number of faces needed to cover $S \cup T$ in a given plane
  graph.
\end{restatable}


\Cref{thm:face-cover-param} provides the first polynomial time algorithm for \probCutUncut on planar graphs for non-singleton $S$ and $T$.
Duan and Xu~\cite{duan2014connectivitypreserving} showed how to solve \probCutUncut on planar graphs for~$|S|=1$ and $|T|=2$.
This was later extended by Bezáková and Langley~\cite{bezakova2014minimumplanar}, who present an~$O(n^4)$-time algorithm for~$|S|=1$ and arbitrary~$T$ on planar graphs.
However, the polynomial time solvability of the case~$|S|=|T|=2$
(which is a generalization of \probND) remained open.


The main tool we develop for showing \Cref{thm:face-cover-param} is a new algorithmic result about computing shortest paths in group-labeled graphs.
We believe that this new result is interesting on its own.
The group that we consider is the \emph{Boolean group} $(\mathbb{Z}_2^d,+)$, consisting of length-$d$ binary vectors, where the operation $+$ is the component-wise exclusive or (xor).
Our algorithm finds a shortest $s$-$t$-path in a graph, whose edges are labeled by elements of $(\mathbb{Z}_2^d,+)$
%
%
%
%
such that the sum of the labels assigned to the edges of the path equals a given value. Furthermore, we impose the constraint that the path can visit certain sets of vertices only once. Formally, we consider the following problem.

\begin{XProblem}
  {\probSP}
  \Input & A graph $G$, two vertices $s$ and $t$, an edge labeling function $g\colon E(G)\rightarrow \mathbb{Z}_2^d$, a value $c\in \mathbb{Z}_2^d$, and $p$ sets of vertices $X_1,\ldots,X_p \subseteq V$. \\
  \Prob & Find an $s$-$t$-path $P$ in $G$ that satisfies 
\begin{itemize}
\item[(i)]\phantomsection\label[prop_ind_1]{xorp1} $\sum_{e\in E(P)}g(e)=c$, and
\item[(ii)]\phantomsection\label[prop_ind_2]{xorp2} for each $i\in[p]$, $|V(P)\cap X_i|\leq 1$,
\end{itemize}
and among such paths minimizes the length.
\end{XProblem}


%
%

In Section~\ref{sec:group-labeled-path}, we give an algorithm for
\probSP, that in fact works for general graphs instead of only planar
graphs.  The result is the following theorem.

\begin{restatable}{theorem}{theoremxorpath}\label{thm:group-labeled-path}
\probSP can be solved in $2^{d+p}\cdot (n+m)^{\Oh(1)}$ time by a one-sided error randomized algorithm.
\end{restatable}

We call the problem variant where we replace path by cycle in the above problem definition \probSC.
Observe that \Cref{thm:group-labeled-path} directly implies also an algorithm for solving \probSC within the same running time.

The proof of \Cref{thm:group-labeled-path} is based on enhancing the technique introduced by  Bj{\"{o}}rklund, Husfeldt, and Taslaman~\cite{bjorklund2012shortestcycle} for the \textsc{$T$-cycle} problem.
In \textsc{$T$-cycle}, the task is to find a shortest cycle that visits a list of specified vertices $T \subseteq V(G)$\footnote{This implies that the algorithm can also take a list of edges by subdividing each target edge and add the new vertex to~$T$.}, and Bj{\"{o}}rklund et al. gave a $2^{|T|} n^{\Oh(1)}$-time algorithm for it.
Our algorithm generalizes the algorithm of Bj{\"{o}}rklund et al., because \textsc{$T$-cycle} can be reduced to \probSC with $d = |T|$ and $p = 0$ as follows.
We assign each vertex~$v \in T$ to one dimension of~$\mathbb{Z}_2^d$, and to enforce that the cycle passes through $v$, we add a true twin~$u$ of~$v$ to the graph and assign the edge~$uv$ the vector in $\mathbb{Z}_2^d$ that has $1$ at only the dimension assigned to $v$.
All other edges are assigned the zero vector $\mathbf{0}$.
Clearly, a cycle evaluating to the all-one vector corresponds to a cycle that visits all vertices in $T$.


%

\medskip%
\noindent%
\textbf{Related work.}
Besides the closely related work on \probND{},
\probTwoDCS{}, and \probCutUncut{} that we already have mentioned above, let us briefly go through other relevant work. 

 \probCutUncut  is a special case of \textsc{Multiway Cut-Uncut}, where for a given equivalence relation on the set of terminals, the task is to find a cut (or node-cut) separating terminals according to the relation. This problem is well-studied in parameterized complexity~\cite{ChitnisCHPP16,CyganKLPPSW21,001RS16}. However, all the previous work in parameterized algorithms on \textsc{Multiway Cut-Uncut} was focused on parameterization by the size of the cut. 
\textsc{Multiway Cut} is also one of the closest relatives of our problem. Here for a given set of $k$ terminals, one looks for a minimum number of edges separating \emph{all} terminals. On planar graph, the seminal paper of Dahlhaus et al.  \cite{DahlhausJPSY94}
 provides an algorithm of running time $n^{\Oh(k)}$. Klein and Marx in \cite{KleinM12} improve the running time to 
  $n^{\Oh(\sqrt{k})}$ and Marx in \cite{Marx12} shows that this running time (assuming the Exponential Time Hypothesis (ETH)) is optimal. 
  
  The second part of \Cref{thm:face-cover-param} concerns the parameterization by the number of faces covering the terminal vertices. Such parameterization comes naturally for optimization problems about connecting or separating terminals in planar graphs. In particular, parameterization by the face cover was investigated for  \textsc{Multiway Cut}~\cite{pandey2022planarmultiway}, \textsc{Steiner Tree}~\cite{bern1990fasterexact,kisfaludibak2020nearly}, and various  flow problems 
~\cite{erickson1987send,filtser2019facecover,krauthgamer2019flow}.
  

Our \Cref{thm:group-labeled-path} belongs to the intersection of two areas around paths in graphs. 
The first area is about polynomial time algorithms computing shortest paths in group-labeled graphs~\cite{derigs1985efficient,grotschel1981weakly,kobayashi2017finding}. 
Recently, Iwata and Yamaguchi~\cite{iwata2022findingshortest}  
gave an algorithm for shortest \emph{non-zero} paths in arbitrary
group-labeled paths. However, for our purposes, we need an algorithm computing a shortest path whose labels sum to a \emph{specific} element of the group.


The second area is about \classFPT algorithms for finding paths in graphs satisfying certain properties \cite{DBLP:journals/jcss/BjorklundHKK17,bjorklund2012shortestcycle,DBLP:conf/soda/FominGKSS23,DBLP:conf/icalp/Koutis08,Williams09}. 
As we already mentioned, our algorithm for \probSP could be seen as an extension of the algorithm of Bj{\"{o}}rklund, Husfeldt, and Taslaman~\cite{bjorklund2012shortestcycle} for the $T$-cycle problem to a group labeled setting.

\bigskip%
\noindent%
\textbf{Organization.} 
The remainder of the article is organized as follows.
We start with a general overview of how we achieve our two main results.
We then present some notation and necessary definitions in Section~\ref{sec:preliminaries}.
Afterwards, we we show how to solve \probSP{} in Section~\ref{sec:group-labeled-path}.
In Section~\ref{sec:cut-uncut-param}, we apply this algorithm to develop a (randomized) FPT-time algorithm for \probCutUncut{} parameterized by the minimum number of faces such that each terminal vertex is incident to at least one such face.
Section~\ref{sec:hardness} is devoted to showing that the W[1]-hardness of \probCutUncut{} parameterized by the number of terminals.
Finally, we present two applications of our FPT-time algorithm to generalize known results from the literature in Section~\ref{sec:applications} and conclude with some remaining open problems in Section \ref{sec:conclusion}.



\subsection{Outline of \Cref{thm:group-labeled-path,thm:face-cover-param}}\label{sec:methods}
We outline the proofs of \Cref{thm:group-labeled-path,thm:face-cover-param}.
We first outline the~$2^{|S|+|T|} \cdot n^{\Oh(1)}$-time algorithm for planar \probCutUncut, and then discuss the setting when $S \cup T$ can be covered by at most $r$ faces.

We observe that any optimal solution to \probCutUncut is an (inclusion-wise) minimal cut in the graph $G$.
Our algorithm is based on the correspondence between minimal cuts in a planar graph and cycles in its dual graph (see \Cref{fig:dualgraph}).
In particular, a set of edges $C \subseteq E(G)$ is a cut-set of minimal cut in $G$ if and only if in the dual graph $G^*$,  the corresponding set $C^* \subseteq E(G^*)$ is a cycle.
Now, to translate \probCutUncut into a problem about finding a cycle $C^*$ in $G^*$, we wish to understand, based on $C^*$, whether two terminal vertices $u$ and $v$ are on the same side of the cut $C$ in $G$ or on different sides.
For this, we observe that if $P \subseteq E(G)$ is the set of edges of an (arbitrary) $u$-$v$-path in $G$ and $P^*$ is the corresponding set of edges in $G^*$, then $u$ and $v$ are on different sides of $C$ if and only if $|C^* \cap P^*|$ is odd.

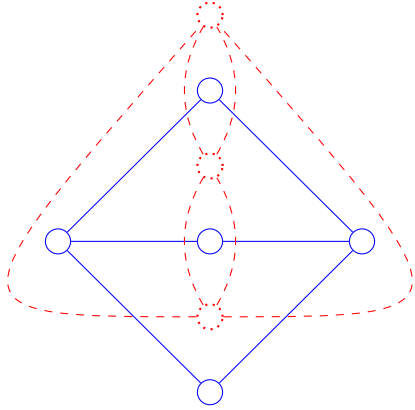
\begin{figure}[t]
\centering
\begin{tikzpicture}
  \tikzstyle{planeblue}=[circle,draw,blue]
  \tikzstyle{planered}=[circle,draw,red,dotted,thick]
  \tikzstyle{planeblueedge}=[blue]
  \tikzstyle{planerededge}=[red,dashed]

  \node[planeblue] at(0,0) (u) {};
  \node[planeblue] at(2,-2) (w1) {} edge[planeblueedge] (u);
  \node[planeblue] at(2,0) (w2) {} edge[planeblueedge] (u);
  \node[planeblue] at(2,2) (w3) {} edge[planeblueedge] (u);
  \node[planeblue] at(4,0) (v) {} edge[planeblueedge] (w1) edge[planeblueedge] (w2) edge[planeblueedge] (w3);

  \node[planered] at(2,3) (a) {};
  \node[planered] at(2,1) (b) {} edge[bend left=30,planerededge] (a) edge[bend right=30,planerededge] (a);
  \node[planered] at(2,-1) (c) {} edge[bend left=30,planerededge] (b) edge[bend right=30,planerededge] (b);
  \draw[planerededge] (c) .. controls (5.5,-1)  ..  (a);
  \draw[planerededge] (c) .. controls (-1.5,-1) .. (a);
\end{tikzpicture}
\caption{An example of a plane graph (blue) and its dual (multi)graph (dashed/red).
  Notice that there are bijections between the faces and the vertices,
  and also between the edges, that is, there is exactly one blue vertex
  in each red face, one red vertex in each blue face, and each red edge
  intersects exactly one blue edge, and vice versa.
}
\label{fig:dualgraph}
\end{figure}

It follows that a constraint stating that $u_i$ and $v_i$ should be on the same/different side of the cut $C$ in $G$ can be expressed as a constraint stating that $|C^* \cap P^*_i|$ should be even/odd for some $P^*_i \subseteq E(G^*)$.
By selecting one vertex $v$ in the set of terminals $S$ and writing a ``same side'' constraint with every other terminal vertex in $S$ and a ``different side'' constraint with every terminal vertex in $T$, the \probCutUncut problem reduces to the problem of finding a shortest cycle $C^*$ in $G^*$ that satisfies $|S|+|T|-1$ given constraints, each requiring that $|C^* \cap P^*_i| \equiv b_i\ (\textrm{mod}\ 2)$ for some $P^*_i \subseteq E(G^*)$ and $b_i \in \{0,1\}$.

This problem can be equivalently phrased as the \probSC problem with $d = |S|+|T|-1$, and therefore \Cref{thm:group-labeled-path} indeed implies a $2^{|S|+|T|}\cdot n^{\Oh(1)}$ time algorithm for \probCutUncut on planar graphs.
Note that here we did not use the condition~\cref{xorp2} in the statement of the \probSP problem; this condition will be used only for the algorithm parameterized by face cover.

We then outline the algorithm of \Cref{thm:group-labeled-path} for \probSP.
This algorithm works not only on planar graphs but also on general graphs, and it is a generalization of the algorithm by Bj{\"{o}}rklund, Husfeldt, and Taslaman~\cite{bjorklund2012shortestcycle} for the $T$-cycle problem.
Our algorithm, like many previous parameterized algorithms for finding paths in graphs~\cite{DBLP:journals/jcss/BjorklundHKK17,bjorklund2012shortestcycle,DBLP:conf/soda/FominGKSS23,DBLP:conf/icalp/Koutis08,Williams09} exploits the cancellation of monomials in polynomials over fields of characteristic two and randomized polynomial identity testing~\cite{Schwartz80,Zippel79}.

The idea of our algorithm is to associate with the input a polynomial over a finite field of characteristic two, and argue that (1) this polynomial is non-zero if and only if a solution exists, and (2) given an assignment of values to variables of the polynomial, the value of the polynomial can be evaluated in time $2^{d+p}\cdot n^{\Oh(1)}$\footnote{Recall that $p$ is the number of constraints for the condition~\cref{xorp2} in the \probSP problem.}.
By the DeMillo--Lipton--Schwartz--Zippel lemma~\cite{Schwartz80,Zippel79}, the problem can then be solved in time $2^{d+p}\cdot n^{\Oh(1)}$ by evaluating the polynomial for a random assignment of values.
Note that solving the decision version also allows to recover the solution by self-reduction.

In more detail, the polynomial associated with the input is defined as follows.
Let us assume that the input graph is a simple graph, as the problem on multigraphs can be easily reduced to simple graphs.
For each edge $e \in E(G)$ of the input graph we associate a variable $f(e)$, and then for an $s$-$t$-walk\footnote{A walk is like a path, but it can contain repeated vertices and edges.} $W = e_1,e_2,\ldots,e_\ell$ of length $\ell$ we associate a monomial $f(W) = \prod_{i=1}^{\ell} f(e_i)$.
Then, for an integer $\ell$, we let $\wf_\ell$ denote the set of all $s$-$t$-walks of length $\ell$ that satisfy the conditions~\cref{xorp1}~and~\cref{xorp2} of the statement of \probSP, and finally let $f(\wf_\ell) = \sum_{W \in \wf_\ell} f(W)$ be the polynomial associated with the input.
As the monomials of $f(\wf_\ell)$ correspond to walks instead of paths, it is not complicated to design a $2^{d+p}\cdot n^{\Oh(1)}$ time dynamic programming algorithm for evaluating the value of $f(\wf_\ell)$.
A more technical part of the proof is to argue that the polynomial $f(\wf_\ell)$ is non-zero if and only if a solution exists, in particular, that monomials corresponding to walks that are not paths cancel each other out.
This argument is a generalization of the argument used by Bj{\"{o}}rklund et al.~\cite{bjorklund2012shortestcycle}.

We then turn to the setting when $S \cup T$ can be covered by $r$ faces in a given plane embedding of $G$. First, we observe that it can be assumed without loss of generality that the input graph is 2-connected. 
This assumption simplifies arguments because the boundary of each face of a plane 2-connected graph is a cycle~\cite{Diestel12}.  Given a plane graph $G$ with terminal sets $S$ and $T$, we can find a minimum face cover of $S\cup T$  by reducing the task to solving an auxiliary instance of the   \textsc{Red-Blue Dominating Set} problem which can be solved in parameterized subexponential time on planar graphs by the results of Alber et al.~\cite{AlberBFKN02}.

Suppose that $f$ is a face of $G$ that covers some terminals and let $C'$ be the cycle forming the frontier of $f$. We use the following crucial observation: for the cut-set  $C \subseteq E(G)$  of any minimal cut in $G$ separating $S$ and $T$, it holds that (i) if $C'$ contains vertices of both sets of terminals, then $C\cap E(C')$ separates $C'$ into two connected components (paths) such that each component contains the vertices of exactly one set of terminals, and (ii) if $C'$ contains vertices of one set, then either $C\cap E(C')=\emptyset$  or $C\cap E(C')$ separates $C'$ into two connected components (paths) such that the terminals are in the same component. We use this observation to restrict the behavior of the cycle $C^*$ in $G^*$ corresponding to a potential solution cut-set $C$. In case (i), we simply delete the edges of $G^*$ that correspond to the edges of $C'$ that should not participate in $C$  (see \Cref{fig:F1}). Case (ii) is more complicated. Suppose that $C'$ contains $q$ terminals.  We find $q$ internally vertex disjoint paths $P_1,\ldots,P_q$ in $C'$ whose end-vertices are the terminals. Then we ``split'' the vertex $f$ of $G^*$ into $q$ vertices $f_1,\ldots,f_q$ in such a way  that each $f_i$ is incident to the edges of $G^*$ corresponding to the edges of $P_i$ (see \Cref{fig:F2}). However, this splitting would allow a cycle in the dual graph to visit the face $f$ several times. To forbid it, we define $X_f=\{f_1,\ldots,f_q\}$ that will be used in constraint~\cref{xorp2} of \probSP and this is the reason why we need constraint~\cref{xorp2} in the problem.

We perform the modifications of  $G^*$ for all the faces in the cover. This allows us to restrict the number of terminals that we should separate. We pick representatives  for each face $f$ in the cover. If the frontier cycle $C'$ of $f$ contains terminals from both sets, we chose one representative from each set from the terminals on $C'$. If $C'$ contains terminals from one set, we choose one representative.  Then we apply the same algorithm as for the parameterization by $|S|+|T|$. The difference is that we work only with the representatives and add constraint~\cref{xorp2} to the auxiliary instance of \probSP given by the sets constructed for the faces from the cover.


\section{Preliminaries}
\label{sec:preliminaries}

For integers $a$ and $b$, we use~$[a,b]$ to denote the set~$\{a, a+1, \ldots, b\}$ and $[b]$ to denote the set $[1, b]$.

\paragraph{Graphs.}
In this paper, we consider undirected multigraphs, that is, we allow multiple edges and self-loops.
We use the standard graph-theoretic notation and refer to~\cite{Diestel12} for undefined notions.
Let~$G=(V,E)$ be an undirected graph. We use $V(G)$ and $E(G)$ to denote the set of vertices and the set of edges of $G$, respectively.
We use~$n$ and~$m$ to denote the number of vertices and edges in~$G$, respectively.
A \emph{path}~$P$ is a graph with vertex set~$\{v_0,v_1,\ldots,v_\ell\}$ and edge set~$\{v_{i-1}v_i \mid i \in [\ell]\}$.
The vertices~$v_0$ and~$v_\ell$ are called the endpoints of~$P$.
A \emph{cycle}~$C$ is a path with an additional edge between the two endpoints.
The length of a path or a cycle in an edge-weighted graph is the sum of weights of its edges and in unweighted graphs the length is the number of edges in it.
For a vertex subset~$U \subseteq V$, we use~$G[U]$ to denote the subgraph of~$G$ induced by the vertices in~$U$ and~$G - U$ to denote~$G[V \setminus V']$.
For a set of edges $S\subseteq E$, we write $G-S$ to denote the graph obtained from $G$ by deleting the edges of $S$.


We are mostly interested in \emph{planar} input graphs.
We refer the reader to the textbooks of Diestel~\cite{Diestel12}
and Agnarsson and Greenlaw~\cite{agnarsson2006graph} for rigorous introductions.
Informally speaking, a graph is planar if it can be drawn on the plane such that its edges do not cross each other.
Such a drawing is called a \emph{planar embedding} of the graph and a planar graph with a planar embedding is called a \emph{plane} graph.
We note that the  planarity check and finding a planar embedding can be done in linear time by the classical algorithm of Hopcroft and Tarjan~\cite{HopcroftT74}.
The \emph{faces} of a plane graph are the regions bounded by a set of edges and that do not contain any other vertices or edges.
The vertices and edges on the boundary of a face form its \emph{frontier}.

Given a plane graph~$G=(V,E)$ with faces~$F$, its \emph{dual} graph~$G^*=(F,E^*)$ (see  \Cref{fig:dualgraph}) is defined as follows.
The vertices of $G^*$ are the faces of~$G$ and for each $e\in E(G)$, $G^*$ has the \emph{dual} edge~$e^*$ whose endpoints are either two faces having~$e$ on their frontiers or~$e^*$ is a self-loop at~$f$ if~$e$ is in the frontier of exactly one face~$f$ (i.e., $e$ is a bridge of $G$).
Observe that $G^*$ is not necessarily simple even if $G$ is a simple graph as the example in \Cref{fig:dualgraph} shows.
We note that $G^*$ is a planar graph that has a plane embedding where each vertex of $G^*$ corresponding to a face $f$ of $G$ is drawn inside $f$ and each dual edge~$e^*$ intersects~$e$ only once and $e^*$ does not intersect any other edge of $G$. Throughout this paper, we assume that $G^*$ has such an embedding.


It is crucial for our results that for  a connected plane graph $G$, each minimal cut in~$G$ has a one-to-one correspondence to a cycle in~$G^*$. To be more precise, recall that each cycle on the plane has exactly two faces.
Then $(A,B)$ is a minimal cut of a plane graph $G$ if and only if there is a cycle $C^*$ in $G^*$ such that the vertices of $A$ are inside one face of $C^*$ and the vertices of $B$ are inside the other face.
Furthermore, $C^*$ is formed by the edges $e^*$ that are dual to the edges $e\in \cut(A)$ and the length of~$C^*$ is~$|\cut(A)|$.

Let $G$ be a plane graph and let $G^*$ be its dual. We say that a path $P$ (a cycle $C$) in $G$ \emph{crosses} a cycle $C^*$ of $G^*$ in~$e\in E(P)$ ($e\in E(C)$, respectively) if $C^*$ contains the edge $e^*\in E^*$ that is dual to~$e$. The \emph{number of crosses} of $P$ and $C^*$ is the number of edges of $P$ where $P$ and $C^*$ cross. 
We use the following observation.

\begin{observation}\label{obs:cycle-sep}
Let~$G$ be a plane graph, let~$s,t\in V$, and let $P$ be an $s$-$t$-path. For any cycle~$C^*$ of~$G^*$, $s$ and $t$ are in distinct faces of $C^*$ if and only if  the number of crosses of $P$ and $C^*$ is odd.
\end{observation}

Lastly, given a subset~$U \subseteq V$ of vertices in a plane graph~$G$ with faces~$F$, a \emph{face cover} of~$U$ is a subset~$F' \subseteq F$ of faces such that each vertex in~$U$ is on the frontier of some face in~$F'$.

\bigskip%
\noindent%
\textbf{Groups.}
The group $(\mathbb{Z}_2^d,+)$ consists of the set of all length-$d$
binary strings, and the sum of two strings is defined as the bitwise
xor of the strings (or addition without carry).  In this regards, it can be
seen as the $d$-dimensional \emph{bitwise xor vector space}
$\mathbb{F}_2^d$.  It is easy to see that this is indeed an (abelian) group:
(1)~The closure property is trivial, since it by definition contains every length-$d$ binary string.
(2)~Associativity can be seen by a simple case analysis, i.e., $(a \oplus b) \oplus c = a \oplus (b \oplus c)$ is bitwise 1 if and only if there is an odd number of 1s in the bit's position.
(3)~The identity element is the all $\mathbf{0}$ vector, i.e.\
$a \oplus \mathbf{0} = a$.
(4)~The inverse element of $a$ is $a$ itself, i.e.,
$a \oplus a = \mathbf{0}$.





\section{Shortest paths under xor constraints}
\label{sec:group-labeled-path}
In the \probSP problem, we are given a graph $G$, two vertices $s$ and $t$, an edge labeling function $g\colon E(G)\rightarrow \mathbb{Z}_2^d$, a value $c\in \mathbb{Z}_2^d$, and $p$ sets of vertices $X_1,\ldots,X_p \subseteq V(G)$.
The problem is to find an $s$-$t$-path $P$ that satisfies (i)\phantomsection\label[prop_ind_1]{xorp3} $\sum_{e \in E(P)} g(e) = c$ and (ii)\phantomsection\label[prop_ind_2]{xorp4} for each $i \in [p]$, $|V(P) \cap X_i| \le 1$, and among such paths minimizes the number of edges in $P$.

In this section we prove \Cref{thm:group-labeled-path}, which we restate here.

\theoremxorpath*




As a corollary, we obtain an algorithm for \probSC{} by guessing one vertex~$v$ in a solution, adding a false twin~$u$ of~$v$ to the input graph such that all edges incident to~$u$ are assinged value zero, and then asking for a shortest~$u$-$v$-path satisfying conditions~\cref{xorp3} and~\cref{xorp4}.

\begin{corollary}\label{cor:group-labeled-cycle}
\probSC can be solved in $2^{d+p}\cdot (n+m)^{\Oh(1)}$ time by a one-sided error randomized algorithm. 
\end{corollary}

\subsection{The algorithm}
In the remainder of this section we assume that the input graph $G$ is a simple graph.
Note that an input $n$-vertex $m$-edge multigraph can be turned into an $(n+m)$-vertex $2m$-edge simple graph by first removing self-loops, and then subdividing each edge once, giving the label of the edge to one of the subdivision edges and labeling the other subdivision edge with zero.
This exactly doubles the length of the solution.
We also assume without loss of generality that $s \neq t$.

Let us next introduce some notation.
We say that a sequence~$(v_0, v_1, \ldots, v_{\ell-1},v_{\ell})$ of~$\ell+1$ vertices is an~$s$-$t$-walk of length $\ell$ if $v_0 = s$, $v_{\ell} = t$, and $v_{i-1} v_{i} \in E(G)$ for each~$i \in [\ell]$.
Note that unlike a path, a walk can contain a vertex more than once.
We say that an $s$-$t$-walk is \emph{feasible} if it satisfies analogies of the contraints~\cref{xorp3} and~\cref{xorp4}, in particular, if
\begin{enumerate}
\item\label{walkcond1} $\sum_{i=1}^{\ell} g(v_{i-1} v_{i}) = c$, and
\item\label{walkcond2} for each $i \in [p]$, there is at most one $j \in [0,\ell]$ such that $v_j \in X_i$.
\end{enumerate}
For an integer $\ell \ge 1$, let $\wf_\ell$ denote the set of all feasible $s$-$t$-walks of length (exactly)~$\ell$.
We associate with~$\wf_\ell$ a polynomial as follows.

Let $q = 2^{\lceil \log_2 n \rceil + 1}$, and recall that $\GF(q)$ is a finite field of characteristic 2 and order $q$.
We define a polynomial over $\GF(q)$ as follows.
For each edge $uv \in E(G)$ we associate a variable $f(uv)$.
Then, for an $s$-$t$-walk $W = (v_0, \ldots, v_{\ell})$ of length $\ell$, we associate the monomial
\begin{align}
f(W) = \prod_{i=1}^{\ell} f(v_{i-1} v_{i})\label{eq:walkmono},
\end{align}
and for the set $\wf_\ell$ of all feasible $s$-$t$-walks of length~$\ell$, we associate the polynomial
\[f(\wf_\ell) = \sum_{W \in \wf_\ell} f(W).\]
Note that the degree of $f(\wf_\ell)$ is $\ell$.
Now, our algorithm will be based on the following lemma, which will be proved in \Cref{subsec:algebr:corr}.

\begin{lemma}
\label{lem:algebr:corr}
The length of the shortest $s$-$t$-path satisfying \cref{xorp3}~and~\cref{xorp4} is equal to the smallest integer $\ell$ such that $f(\wf_\ell)$ is a non-zero polynomial.
If no such $\ell$ exists, then no such $s$-$t$-path exists.
\end{lemma}

Given \Cref{lem:algebr:corr}, it remains to design an algorithm for testing if $f(\wf_\ell)$ is a non-zero polynomial.
For this, we use the DeMillo--Lipton--Schwartz--Zippel lemma.

\begin{lemma}[\cite{Schwartz80,Zippel79}]
\label{lem:schwartzzippel}
Let $p(x_1, \ldots, x_n)$ be a non-zero polynomial of degree $d$ over a field $\mathbb{F}$, and let $S$ be a subset of $\mathbb{F}$.
If each $x_i$ is independently assigned a uniformly random value from $S$, then~$p(x_1, \ldots, x_n) = 0$ with probability at most $d/|S|$.
\end{lemma}

By \Cref{lem:schwartzzippel} to probabilistically test if $f(\wf_\ell)$ is non-zero it suffices to evaluate $f(\wf_\ell)$ on a random assignment of values from~$\GF(q)$ to the variables $f(uv)$.
Because the degree of $f(\wf_\ell)$ is $\ell \le n$ and the order of~$\GF(q)$ is $q \ge 2n$, this test is correct with probability at least $0.5$ whenever $f(\wf_\ell)$ is non-zero.
Note that if $f(\wf_\ell)$ is the zero polynomial, this test is always correct.
Next we show that this evaluation can be done efficiently.

\begin{lemma}
\label{lem:algalgeval}
Given an assignment of values to the variables $f(uv)$ for all $uv \in E(G)$, the value of the polynomial $f(\wf_\ell)$ can be evaluated in time $\Oh(2^{d+p} n^2 \ell)$.
\end{lemma}
\begin{proof}
We evaluate the polynomial by dynamic programming on walks.
For $u \in V(G)$, $l \in [0,\ell]$, $y \in \mathbb{Z}_2^d$, and $T \subseteq [p]$, let us denote by $\wf(u, l, y, T)$ the set of $s$-$u$-walks $(s = v_0, v_1, \ldots, v_{l} = u)$ of length~$l$ that have 
\begin{itemize}
\item $\sum_{i=1}^{l} g(v_{i-1} v_{i}) = y$, 
\item for each $i \in [p] \setminus T$, it holds that $\{v_0,v_1, \ldots, v_l\} \cap X_i = \emptyset$, and
\item for each $i \in T$, there exists exactly one $j \in [0,l]$ so that $v_j \in X_i$.
\end{itemize}
Then, we denote by $f(\wf(u, l, y, T))$ the value $\sum_{W \in \wf(u, l, y, T)} f(W)$, where $f(W)$ is defined as in \Cref{eq:walkmono}, with the empty product interpreted as being equal to $1$.
Now, we have that $f(\wf_\ell) = \sum_{T \subseteq [p]} f(\wf(t, \ell, c, T))$.
It remains to show that the values $f(\wf(u, l, y, T))$ can be computed by dynamic programming.

Let~$T_v = \{i \in [p] \mid v \in X_i\}$ for each~$v \in V(G)$.
Then, the values for $l = 0$ are computed by setting $f(\wf(s, 0, 0, T_s)) = 1$ and all other values with $l=0$ to $0$.
Then, when $l \ge 1$, the values $f(\wf(u, l, y, T))$ are computed by dynamic programming from the values for smaller $l$ as follows.
\begin{itemize}
\item If $T_u \subseteq T$, then 
$f(\wf(u, l, y, T)) = \sum\limits_{uw \in E(G)} f(uw) \cdot f(\wf(w, l-1, y-g(\{u,w\}), T \setminus T_u))$.
\item Otherwise, $f(\wf(u, l, y, T)) = 0$.
\end{itemize}
This clearly computes the values correctly, and runs in overall~$\Oh(2^{d+p} n^2 \ell)$~time. 
\end{proof}
     
Now, our algorithm works by using \Cref{lem:algalgeval} to evaluate $f(\wf_\ell)$ for random assignments of values to variables $f(uv)$ for increasing values of $\ell \le n$, and once it evaluates to non-zero, reports that $\ell$ is the length of the shortest $s$-$t$-path satisfying \cref{xorp3}~and~\cref{xorp4}.
If no such~$\ell \le n$ is found, the algorithm reports that no such $s$-$t$-path exists. 
Note that the correctness of the algorithm depends only on the randomness on the evaluation with the correct $\ell$, and therefore the algorithm is correct with probability at least~$0.5$, and never reports a length shorter than the length of a shortest solution.
This probability can be exponentially improved by running the algorithm multiple times.
To recover the solution, it suffices to use the algorithm to test which edges can be removed from the graph $G$ until $G$ turns into an $s$-$t$-path.
Clearly, to both recover the solution and to have an exponentially small error probability it suffices to run the algorithm a polynomial number of times, so this finishes the proof of \Cref{thm:group-labeled-path}, modulo the proof of \Cref{lem:algebr:corr} that will be given in the next subsection.

\subsection{Proof of correctness}
\label{subsec:algebr:corr}
This section is devoted to the proof of \Cref{lem:algebr:corr}.
We first prove the direction that the existence of a solution of length $\ell$ implies that $f(\wf_\ell)$ is non-zero.

\begin{lemma}
\label{lem:algebralgcorrfpfi:forw}
If an $s$-$t$-path of length $\ell$ satisfying \cref{xorp3}~and~\cref{xorp4} exists, then $f(\wf_\ell)$ is a non-zero polynomial.
\end{lemma}
\begin{proof}
Let $W=(s = v_0, v_1, \ldots, v_{\ell} = t)$ be the sequence of vertices on an $s$-$t$-path of length~$\ell$ satisfying conditions~\cref{xorp3}~and~\cref{xorp4}.
Note that~$W$ is a feasible $s$-$t$-walk and~$W \in \wf_\ell$.
Because each vertex occurs in the walk~$W$ at most once, we observe that $W$ can be determined uniquely from its set of edges, and therefore $W$ is the only walk in~$\wf_\ell$ with the monomial $f(W) = \prod_{i=1}^{\ell} f(v_{i-1} v_{i})$.
Thus, the monomial~$f(W)$ occurs in the polynomial~$f(\wf_\ell)$ with coefficient $1$, and therefore $f(\wf_\ell)$ is non-zero.
\end{proof}

It remains to prove that if no solutions of length at most $\ell$ exists, then $f(\wf_\ell)$ is the zero polynomial.
For this, let us state our main lemma, but delay its proof until the end of this subsection.

\begin{lemma}
\label{lem:algebralgcorrfpfi}
If no $s$-$t$-path of length at most~$\ell$ satisfying conditions~\cref{xorp3}~and~\cref{xorp4} exists, then there exists a function~$\phi : \wf_\ell \rightarrow \wf_\ell$ such that for every $W \in \wf_\ell$ it holds that
\begin{enumerate}
\item \label{lem:algebralgcorrfpfi:pinvo} $\phi(\phi(W)) = W$,
\item \label{lem:algebralgcorrfpfi:fpf} $\phi(W) \neq W$, and
\item \label{lem:algebralgcorrfpfi:pf} $f(\phi(W)) = f(W)$.
\end{enumerate}
\end{lemma}

Now, assuming \Cref{lem:algebralgcorrfpfi}, the proof of \Cref{lem:algebr:corr} can be finished as follows.

\begin{lemma}
\label{lem:algebralgcorrfpfi:backw}
If no $s$-$t$-path of length at most $\ell$ satisfying conditions~\cref{xorp3}~and~\cref{xorp4} exists, then $f(\wf_\ell)$ is the zero polynomial.
\end{lemma}
\begin{proof}
Let $\phi$ be the function given by \Cref{lem:algebralgcorrfpfi}.
By properties \ref{lem:algebralgcorrfpfi:pinvo}~and~\ref{lem:algebralgcorrfpfi:fpf}, the set $\wf_\ell$ can be partitioned into pairs $\{W, \phi(W)\}$.
Now, property~\ref{lem:algebralgcorrfpfi:pf} states that~$f(W) = f(\phi(W))$ and since~$\GF(q)$ is a field of characteristic~2, it holds that $f(W) + f(\phi(W)) = 0$.
Thus, $\sum_{W \in \wf_\ell} f(W) = 0$.
\end{proof}

Putting \Cref{lem:algebralgcorrfpfi:forw,lem:algebralgcorrfpfi:backw} together implies \Cref{lem:algebr:corr}.
It remains to prove \Cref{lem:algebralgcorrfpfi}.

\begin{proof}[Proof of \Cref{lem:algebralgcorrfpfi}]
Assume that no $s$-$t$-path of length at most $\ell$ satisfying \cref{xorp3}~and~\cref{xorp4} exists.
We will define the function $\phi : \wf_\ell \rightarrow \wf_\ell$ explicitly and show that it satisfies all of the required properties.
Let $W = (v_0,v_1, \ldots, v_{\ell})$ be an $s$-$t$-walk in $\wf_\ell$.
The idea of the definition of $\phi$ will be to locate a \emph{subwalk}~$(v_i, v_{i+1}, \ldots, v_{j-1}, v_j)$ of~$W$ where $0 \le i < j \le \ell$, and reverse the subwalk, i.e., map the walk 
\[W = (v_0,v_1,\ldots,v_{i-1},v_i,v_{i+1},\ldots,v_{j-1},v_j,v_{j+1}, \ldots, v_{\ell})\]
into the walk
\[W\overleftarrow{[i,j]} = (v_0,v_1, \ldots, v_{i-1}, v_j, v_{j-1}, \ldots, v_{i+1}, v_i, v_{j+1}, \ldots, v_{\ell}).\]
In particular, we will have that $\phi(W) = W\overleftarrow{[i,j]}$ for a carefully chosen pair $i,j$ with $0 \le i < j \le \ell$.
This pair will be chosen so that $v_i = v_j$, which ensures that~$W\overleftarrow{[i,j]} \in \wf_\ell$ and $f(W\overleftarrow{[i,j]}) = f(W)$ since the multiset of pairs of adjacent vertices in the walk does not change.

It remains to argue that such a pair $i,j$ can be chosen so that the properties $\phi(\phi(W)) = W$ and~$\phi(W) \neq W$ hold.
Observe that the property $\phi(W) \neq W$ holds if and only if the subwalk from $i$ to~$j$ is not a \emph{palindrome}, i.e., a sequence that is the same when reversed.
Now we define a process that outputs a pair $i,j$ so that $0 \le i<j \le \ell$, $v_i = v_j$, and the subwalk from $i$ to~$j$ is not a palindrome.

The process starts by setting $i=j=0$.
Then, it repeats the following: It first selects $i$ to be the smallest integer $i>j$ so that the vertex $v_i$ occurs in the walk in the indices greater than $j$ more than once.
If no such $i$ exists, it outputs $\textsf{FAIL}$.
Then, it sets $j$ to be the largest integer so that $v_i = v_j$, in particular, the index of the last occurrence of $v_i$ in the walk.
At this point it is guaranteed that $0 \le i < j \le \ell$ and $v_i = v_j$.
Now, if the subwalk from $i$ to $j$ is not a palindrome, it outputs the pair $i,j$.
Otherwise, the process repeats.

Observe that the process always outputs either \textsf{FAIL} or a pair~$i,j$ with~$0 \le i < j \le \ell$ and~$v_i = v_j$ such that the subwalk from~$i$ to~$j$ is not a palindrome.
We prove that it actually never outputs \textsf{FAIL}.
\begin{claim}
The process defined above never outputs \textsf{FAIL}.
\end{claim}
\begin{proof}[Proof of claim]
Suppose that the process outputted \textsf{FAIL}, and let $i_1 < j_1 < i_2 < j_2 < \ldots < i_t < j_t$ be the sequence of pairs $i,j$ considered during the process.
We define the \emph{contracted walk} $W'$ to be the subsequence of $W = (v_0, \ldots, v_{\ell})$ obtained by removing the vertices on the indices in $[i_1+1, j_1] \cup [i_2+1, j_2] \cup \ldots \cup [i_t+1, j_t]$ from $W$.
In particular, $W'$ is obtained from $W$ by contracting each palindrome $v_{i_k}, \ldots, v_{j_k}$ considered in the process into a single vertex $v_{i_k}$.

Now, we claim that $W'$ is a feasible $s$-$t$-walk of length at most $\ell$, and moreover that no vertex occurs more than once in $W'$.
This is a contradiction, because in that case $W'$ would be in fact an $s$-$t$-path of length at most $\ell$ that satisfies conditions~\cref{xorp3}~and~\cref{xorp4}, but assumed that no such $s$-$t$-path exists.
We observe that the contracted walk $W'$ is indeed an $s$-$t$-walk, because it was obtained from an $s$-$t$-walk by contracting subwalks that each start and end in a same vertex.
It also clearly has length at most $\ell$, and it satisfies the condition~\cref{xorp4} because the multiset of vertices in $W'$ is a subset of the multiset of vertices in $W$.
For condition~\cref{xorp3}, we observe that if a subwalk $v_i, \ldots, v_j$ is a palindrome, then $\sum_{k=i+1}^{j} g(v_{k-1} v_{k}) = 0$, because each pair of adjacent vertices occurs an even number of times and we are working in the group~$\mathbb{Z}_2^d$.
Thus, contracting the palindromes does not change the sum of the edge labels on $W$, and thus $W'$ satisfies condition~\cref{xorp4}.

Lastly, we argue that no vertex occurs more than once in $W'$.
For the sake of contradiction, suppose that some vertex occurs more than once in $W'$, which in particular implies that there are indices $i'$,$j'$ with~$0 \le i' < j' \le \ell$ and~$v_{i'} = v_{j'}$ that are not in~$[i_1+1, j_1] \cup [i_2+1, j_2] \cup \ldots \cup [i_t+1, j_t]$.
If~$i' < i_1$, then this would contradict the choice of $i_1$, and if $i' = i_1$, then this would contradict the choice of~$j_1$.
Similarly, if $j_k < i' \le i_{k+1}$ for some~$1 \le k < t$, then this would contradict either the choice of $i_{k+1}$ or $j_{k+1}$, and if $i' > j_t$, then this would contradict the fact that $i_t,j_t$ was the last pair considered by the algorithm.
\end{proof}

Now, the function $\phi : \wf_\ell \rightarrow \wf_\ell$ is defined as $\phi(W) = W\overleftarrow{[i,j]}$, where $i,j$ is the pair outputted by the process described above.
We have already proved that $\phi(W) \neq W$ and $f(\phi(W)) = f(W)$, so it remains to prove that $\phi(\phi(W)) = W$.
For this, it remains to observe that the operation $W\overleftarrow{[i,j]}$ does not change how the process for selecting $i$,$j$ behaves, because it does not change the walk before the index $i$ and it does not change the fact that the last occurrence of $v_i$ is at the index $j$.
\end{proof}

This finishes the proof of \Cref{thm:group-labeled-path}.


\newcommand{\ff}[1]{\ensuremath{f}}
\newcommand{\fff}[2]{\ensuremath{f_{{#2}}}}
\section{Two-Sets Cut-Uncut parameterized by face cover}
\label{sec:cut-uncut-param}

In this section, we show that \probCutUncut is \classFPT when parameterized by the minimum number of faces in a plane embedding of the input graph covering the terminals. We use the following crucial observations.
First, we observe that a minimum face cover can be found in \classFPT time when parameterized by the size of a cover. For this, we use the known results about the \textsc{Red-Blue Dominating Set} problem on planar graphs. The task of this problem is, given a bipartite graph $G$ whose vertices are partitioned into two sets $R$ and $B$ (\emph{red} and \emph{blue} vertices, respectively) and an integer~$r\geq 1$, to decide whether there is a set $D$ of at most~$r$ red vertices that dominates the blue vertices, that is, each $v\in B$ is adjacent to at least one vertex of $D$. It was proved by Alber et al.~\cite{AlberBFKN02} that this problem can be solved in~$2^{\Oh(\sqrt{r})}\cdot n$ time on planar graphs. 

\begin{lemma}\label{lem:cover}
It can be decided in $2^{\Oh(\sqrt{r})}\cdot n^{\Oh(1)}$ time whether a set of vertices $U$ of a plane graph $G$ has a face cover of size at most $r$. Furthermore, if such a cover exists, it can be found in the same time.
\end{lemma}

\begin{proof}
Let $G$ be a plane graph with faces~$F$ and let $U\subseteq V(G)$. We construct the  instance of \textsc{Red-Blue Dominating Set} by setting $R=F$, $B=U$, and making $f\in R$ adjacent to $v\in U$ if $v$ is in the frontier of the face~$f$. Then, $U$ can be covered by at most~$r$ faces if and only if the set of blue vertices can be dominated by at most $r$ red vertices. Note that the constructed red-blue graph is planar.
Then the results of Alber et al.~\cite{AlberBFKN02} imply the claim of the lemma.
\end{proof}

Second, we note the following connection.

\begin{observation}\label{obs:mincut}
Let $(G,S,T,k)$ be an instance of \probCutUncut  where $G$ is a connected graph. Then $(G,S,T,k)$ is a yes-instance if and only if there is a minimal cut $(A,B)$ of $G$ with $\left| \cut(A) \right| \leq k$ such that $S\subseteq A$ and $T\subseteq B$.
\end{observation}

\begin{proof}
If $G$ has a minimal cut $(A,B)$ with $\left| \cut(A) \right| \leq k$ such that $S\subseteq A$ and $T\subseteq B$, then~$(G,S,T,k)$ is a yes-instance because $G[A]$ and $G[B]$ are connected. For the opposite direction, let 
 $(A,B)$ be a cut of $G$ of minimum size  such that the vertices of $S$ are in the same connected component of $G[A]$ and the vertices of $T$ are in the same connected component of $G[B]$. We claim that $(A,B)$ is minimal.
 Assume towards a contradiction that~$(A,B)$ is not minimal, that is, $G[A]$ or $G[B]$ is disconnected.
 We assume without loss of generality that $G[B]$ is disconnected. Then, there is a subset~$R\subseteq B$ such that $G[R]$ is a connected component of $G[B]$ and $R\cap T=\emptyset$. Consider the cut $(A',B')$ where $A'=A\cup R$ and $B'=B\setminus R$. We have that the vertices of $S$ are in the same connected component of $G[A']$ and the vertices of $T$ are in the same connected component of $G[B']$. However, because $G$ is connected, $G$ has an edge with one endpoint in $A$ and the other in $R$.
Thus, $\left| \cut(A') \right| < \left| \cut(A) \right|$ contradicting the choice of $(A,B)$. This proves the observation.
\end{proof}

This observation implies that to solve \probCutUncut  on a plane graph $G$ , we have to find a shortest cycle $C^*$ in the dual graph $G^*$ such that the vertices of $S$ and $T$ are in distinct faces of $C^*$.
First, we observe that we can assume without loss of generality that the input graph is 2-connected. This assumption simplifies arguments because the frontier of each face of a plane 2-connected graph is a cycle~\cite{Diestel12}. 

\begin{lemma}\label{lem:two-conn}
There is a polynomial-time algorithm that, given an instance $(G,S,T,k)$ of \probCutUncut, either solves the problem or outputs an equivalent instance $(G',S',T',k)$ of \probCutUncut where $G'$ is 
a $2$-connected induced subgraph of~$G$. Furthermore, given a planar embedding of~$G$ such that $S\cup T$ can be covered by at most $r$ faces,  $S'\cup T'$ can be covered in the induced embedding of~$G'$ by at most $r$ faces.  
\end{lemma}

\begin{proof} 
Suppose that $G$ is disconnected. Then, $(G,S,T,k)$ is a trivial no-instance if either  vertices of $S$ occur in distinct connected components of $G$ or, symmetrically, vertices of~$T$ occur in distinct connected components. 
Also, we have that $(G,S,T,k)$ is a trivial yes-instance if the vertices $S$ and the vertices of~$T$ are in distinct connected components of $G$. It remains the case when the vertices of~$S$ and~$T$ are in the same connected component of $G$.
Then, we can remove all other connected components. Hence, we may assume that~$G$ is connected.

The claim of the lemma is trivial if $G$ is 2-connected. Assume that $G$ has a cut-vertex $v$. Then, there is a separation $(X,Y)$ of $G$ with the separator $v$, that is, there 
are subsets~$X,Y \subseteq V(G)$ such that $X\cup Y=V(G)$, $X\cap Y=\{v\}$, and~$G$ has no edge~$xy$ with $x\in X\setminus Y$ and $y\in Y\setminus X$. Observe that if $(A,B)$ is a minimal cut of $G$ with $v\in A$, then either $X\subseteq A$ or $Y\subseteq B$ because, otherwise, $G[B]$ is disconnected. Thus, for any minimal cut $(A,B)$ of $G$, $\cut(A)=\cut(A')$ where~$(A',B')$ is a minimal cut of either $G[X]$ or $G[Y]$.
This observation leads to the following four cases (up to the symmetry). 
First, we consider the case when all the terminals are in one part of the separation.

\subparagraph{Case~1.} $(X\setminus Y)\cap (S\cup T)=\emptyset$. Then, for any minimal cut $(A,B)$ of $G$ with $v\in A$ such that~$S$ and~$T$ are in distinct parts, $X\subseteq A$ and~$(G,S,T,k)$ is equivalent to $(G[Y],S,T,k)$, that is, we can reduce the input instance by deleting the vertices of $X\setminus Y$.

\medskip
Next, we assume that each part of the separation contains terminals from both $S$ and $T$.

\subparagraph{Case~2.} $X\cap S\neq\emptyset$, $X\cap T\neq\emptyset$, $Y\cap S\neq\emptyset$, and $Y\cap T\neq\emptyset$. Because for any minimal cut $(A,B)$ of $G$ with $v\in A$, either $X\subseteq A$ or $Y\subseteq B$, we have that there is no minimal cut separating $S$ and $T$. 
Thus, $(G,S,T)$ is a no-instance. We report that there is no solution and stop.

\medskip
Now, we assume that one set of the separation contains terminals from both $S$ and $T$ and the other includes terminals only from one of the sets $S$ and $T$.  

\subparagraph{Case~3.}  $(X\setminus Y)\cap S\neq\emptyset$, $X\cap T=\emptyset$, $Y\cap S\neq\emptyset$, and $Y\cap T\neq\emptyset$. Then, for any minimal cut $(A,B)$ of $G$ with $v\in A$ such that $S$ and $T$ are in distinct parts,  $X\subseteq A$. Furthermore, $S\subseteq A$ and $T\subseteq B$.  This implies that $(G,S,T,k)$ is equivalent to $(G[Y],S',T,k)$ where $S'=(S\setminus X)\cup \{v\}$. This allows us to reduce the input instance by deleting the vertices of $X\setminus Y$ and modifying $S$.  

\medskip Finally, we consider the case when each set of the separation contains terminals only from one of the sets $S$ and $T$.

\subparagraph{Case~4.}  $(X\setminus Y)\cap S\neq\emptyset$, $X\cap T=\emptyset$, $Y\cap S=\emptyset$, and $(Y\setminus X)\cap T\neq\emptyset$. Consider a minimal cut~$(A,B)$ of $G$ with $v\in A$ such that $S$ and $T$ are in different parts. If $X\subseteq A$, then $\cut(A)=\cut(A')$ for a minimal cut $(A',B')$ of $G[B]$ such that $\{v\}\in A$ and $T\subseteq B$. The case $Y\subseteq A$ is symmetric. This allows us to solve the problem in polynomial time by making use of the result of Bez\'akov\'a and Langley~\cite{bezakova2014minimumplanar}.
By this result, the problem when $|S|=1$ is solvable in polynomial time. 
 We apply the algorithm of  Bez\'akov\'a and Langley for $(G[Y],\{v\},T,k)$ and $(G[X],S,\{v\},k)$ and conclude that~$(G,S,T,k)$ is a yes-instance if at least one of these instances is a yes-instance. 

\medskip
Because the cut-vertices of~$G$ can be listed in linear time by the classical algorithm by Tarjan~\cite{DBLP:journals/siamcomp/Tarjan72}, we conclude that in a polynomial time, we either solve the problem or reduce the input instance to an equivalent instance $(G',S',T',k)$ where $G'$ has no cut-vertices. If the obtained graph~$G'$ has two vertices, the problem is trivial because $|S'|=|T'|=1$ and the minimum cut is unique. Otherwise, $G'$
is a $2$-connected induced subgraph of $G$. 

\medskip
To show the second claim of the lemma, assume that $G$ is embedded on the plane and $S\cup T$ is covered by a set of faces $F'$. Consider the induced embedding of $G'$.  Because $G'$ is an induced subgraph of $G$, for each face $f$ of $G$, there is a face $f'$ of $G'$ such that each vertex covered by~$f$ in~$G$ is covered by~$f'$ in~$G'$. We construct the set of faces $F''$ of $G'$ from $F'$ by including each such face~$f'$ for each face~$f\in F'$. We claim that $F''$ covers $S'\cup T'$. By the definition of~$F'$, we have that every~$x\in (S\cup T)\cap (S'\cup T')$ is covered by $F''$. Hence, we have to prove that every~$x\in (S'\cup T')\setminus (S\cup T)$ is covered as well. Note that new terminals are introduced only in Case~3.  Assume that $(X\setminus Y)\cap S\neq\emptyset$, $X\cap T=\emptyset$, $Y\cap S\neq\emptyset$, and~$Y\cap T\neq\emptyset$. Then, there is a face~$f\in F'$ that covers a vertex of $(X\setminus Y)\cap S$. Since~$v$ is a cut-vertex, we have that the face~$f'$ covers~$v$. This implies that $F''$ covers $S'\cup T$ where $S'=(S\setminus X)\cup \{v\}$. This concludes the proof.
\end{proof}

From now on, we assume that the graph of the considered instances of  \probCutUncut is $2$-connected. We remind that the frontier of each face of a plane  2-connected graph $G$ is a cycle. Moreover, 
the dual graph $G^*$ has no loops.  Also, since loops are irrelevant for \probCutUncut, we assume that the input graph has no loops. 

%

We use the following separation properties for vertices on the frontier of the same face of a graph.

\begin{lemma}\label{lem:sepsets}
Let $G$ be a plane graph  and let $X$ and $Y$ be disjoint nonempty sets of vertices of the cycle $C$ which forms the frontier of a face $f$ of $G$. Let~$C^*$ be any cycle in~$G^*$. Then the vertices of $X$ and the vertices of $Y$ are in distinct faces of $C^*$ if and only if $f\in V(C^*)$ and $C$ crosses $C^*$ in two edges~$e_1$ and~$e_2$ such that (i) the vertices of $X$ are in the same connected component of $C-\{e_1,e_2\}$, (ii)  the vertices of~$Y$ are in the same connected component of $C-\{e_1,e_2\}$, and (iii) the vertices of $X$ and the vertices of $Y$ are in distinct connected components of $C-\{e_1,e_2\}$.
\end{lemma}

\begin{proof}
Suppose that  the vertices of $X$ and the vertices of $Y$ are in distinct faces of $C^*$. Then, $f$ is a vertex of $C^*$. Let $e_1^*$ and $e_2^*$ be the edges of  $C^*$ incident to $f$ and let~$e_1$ and~$e_2$ be the dual edges of $e_1^*$ and $e_2^*$, respectively. Note that~$C$ contains both~$e_1$ and~$e_2$ and~$C$ crosses $C^*$ only in these two edges. We have that $C-\{e_1,e_2\}$ has two connected components $P_1$ and $P_2$ that are paths. Since~$C^*$ separates $X$ and $Y$, we have that~$X$ is fully contained in~$P_1$ or fully contained in~$P_2$ and~$Y$ is fully contained in the respective other path. Thus, conditions~(i)--(iii) are fulfilled. 

For the opposite direction, assume that  $C$ crosses $C^*$ in two edges $e_1$ and $e_2$ such that conditions~(i)--(iii) are fulfilled. Consider an $x$-$y$-path $P$ in $C$ for arbitrary $x\in X$ and $y\in Y$ containing $e_1$ and excluding $e_2$ that exists by (i)--(iii). The number of crosses of $P$ and~$C^*$ is one. Hence, $x$ and $y$ are in distinct faces of~$C^*$ by \Cref{obs:cycle-sep}. This concludes the proof. 
\end{proof}

\begin{lemma}\label{lem:nonsepset}
Let $G$ be a plane graph  and let $X$ be a non-empty set of vertices of the cycle $C$ forming the frontier of a face $f$ of $G$. Let~$C^*$ be any cycle in~$G^*$. Then, the vertices of~$X$ are in the same face of $C^*$ if and only if either $f\notin V(C^*)$ and $C$ does not cross $C^*$ or $f\in V(C^*)$ and $C$ crosses $C^*$ in two edges $e_1$ and $e_2$ such that the vertices of $X$ are in the same connected component of $C-\{e_1,e_2\}$.
\end{lemma}

\begin{proof}
Suppose that the vertices of~$X$ are in the same face of  $C^*$ and assume that $C$ crosses $C^*$. 
Then, $f$ is a vertex of $C^*$ and $C^*$ has two edges~$e_1^*$ and~$e_2^*$ incident to $f$. 
We have that $C$ crosses $C^*$ in the edges~$e_1$ and~$e_2$ that are dual to~$e_1^*$ and~$e_2^*$, respectively. Since~$C$ is a cycle,  $C-\{e_1,e_2\}$ has two connected components $P_1$ and $P_2$ that are both paths. We show that either $X\subseteq V(P_1)$ or $X\subseteq V(P_2)$. For the sake of contradiction, assume that there are $x,y\in X$ such that $x\in V(P_1)$ and $y\in V(P_2)$. Then, there is an $x$-$y$-path~$P$ in~$C$ that contains $e_1$ but excludes $e_2$. The number of crosses of $P$ and $C^*$ is one and $x$ and $y$ are therefore in distinct faces of~$C^*$ by \Cref{obs:cycle-sep}; a contradiction. Hence, the vertices of $X$ are in the same connected component of $C-\{e_1,e_2\}$. 

For the opposite direction, assume that either $C$ does not cross $C^*$ or $C$ crosses $C^*$ in two edges $e_1$ and $e_2$ such that the vertices of $X$ are in the same connected component of $C-\{e_1,e_2\}$. In both cases,  for any two vertices $x,y\in X$, there is an $x$-$y$-path $P$ that does not cross $C^*$. Then by \Cref{obs:cycle-sep}, the vertices  of $X$ are in the same face of $C^*$. This concludes the proof.
\end{proof}

%
%

We are now in a position to present the main result of this section.

\maintheoremfacecover*

\begin{proof}
We show the claim for the parameterization by the size of a face cover of the terminals and then explain how a simplified version of the algorithm can be used for the parameterization by the number of terminals.

Let $(G,S,T,k)$ be an instance of \probCutUncut where we are given an embedding of $G$ on the plane. We remind that $G$ is assumed to be $2$-connected. We use the embedding of $G$ to construct the dual graph $G^*$ together with its embedding. By \Cref{obs:mincut}, our task is to find a cycle $C^*$ in $G^*$ of length at most $k$ such that $S$ and $T$ are in distinct faces of $C^*$. We find such a cycle using the algorithm for \probSC from \Cref{cor:group-labeled-cycle}.

We use \Cref{lem:cover} to verify whether there is a set of faces $F'$ of size at most $r$ that cover $S\cup T$. If such a cover does not exist, we stop.  From now on, we assume that $F'$ is given.
We partition $F'$ into two sets where $F_1\subseteq F'$ is the set of faces having vertices from both $S$ and $T$ on their frontiers and $F_2\subseteq F'$ consists of the faces $f\in F'$ such that the frontier of $f$ contains ether only vertices of $S$ or only vertices of~$T$.   We modify $G^*$ by analyzing each face $f\in F'$. The ultimate aim of the modification is to reduce the number of considered terminals.  

\begin{figure}[t]
\centering
\scalebox{0.7}{
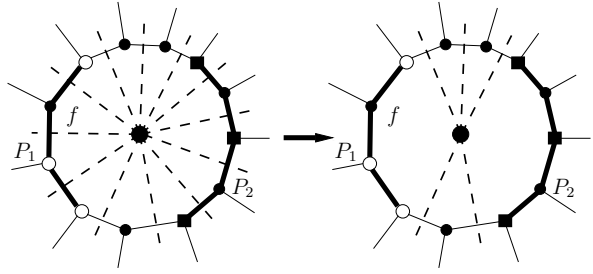}
\caption{The modification for $f\in F_1$. The vertices of $S'$ are shown by white circles, the vertices of $T'$ are shown by black squares, and the other vertices of $G$ are shown by black circles. The edges of $G^*$ are shown by dashed lines. The paths $P_1$ and $P_2$ are shown by thick lines.}
\label{fig:F1}
\end{figure}


\subparagraph{Modifications for $F_1$.} Let $f\in F_1$ and let $C$ be the cycle of $G$ forming the frontier of~$f$. Recall that $S'=S\cap V(C)\neq \emptyset$ and $T'=T\cap V(C)\neq\emptyset$. If there are no two edges $e_1,e_2\in E(C)$ such that the vertices of $S'$ and $T'$ are in distinct connected components of $C-\{e_1,e_2\}$, then by \Cref{lem:sepsets}, there is no cycle $C^*$ such that the vertices of $S'$ and the vertices of $T'$ are in distinct faces of $C^*$. This implies  that~$(G,S,T,k)$ is a no-instance. Hence, we assume that this is not the case and select two inclusion-minimal disjoint paths $P_1$ and $P_2$ in $C$ such that $S'\subseteq V(P_1)$ and $T'\subseteq V(P_2)$. We modify $G^*$ by deleting each edge $e^*$ incident to $f$ that is dual to an edge $e\in E(P_1)\cup E(P_2)$ (see \Cref{fig:F1}).

\begin{figure}[t]
\centering
\scalebox{0.7}{
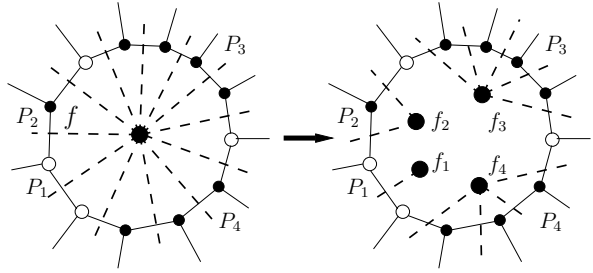}
\caption{The modification for $\ff{i} \in F_2$. The vertices of $R$ are shown by white circles and the other vertices of $G$ are shown by small black circles. The vertex $\ff{i}$ of $G^*$ and the vertices $\fff{i}{1},\fff{i}{2}\dots,\fff{i}{4}$ are shown by large black circles and the edges of $G^*$ and the constructed new edges are shown by dashed lines. }
\label{fig:F2}
\end{figure}

\subparagraph{Modifications for $F_2$.} Let $\ff{i} \in F_2$, let~$C$ be the cycle of $G$ forming the frontier of~$\ff{i}$, and let $L=V(C) \cap (S\cup T)$. Note that by definition of~$F_2$, either $L\subseteq S$ or $L\subseteq T$. We split the vertex~$\ff{i}$ of $G^*$ into $q=|L|$ vertices~$\fff{i}{1},\fff{i}{2},\ldots,\fff{i}{q}$ as follows. If $q\geq 2$, then $C$ contains~$q$ internally vertex disjoint paths~$P_1,P_2,\ldots,P_q$ whose end-vertices are in $L$ (and whose internal vertices are not in~$L$).
We then 
\begin{itemize}
\item delete $\ff{i}$ and construct a set $X_f=\{\fff{i}{1},\fff{i}{2},\ldots,\fff{i}{q}\}$ of $q$ new vertices,
\item for each $j\in[q]$ and edge~$e$ in~$P_j$, we replace the dual edge $e^*$ of $G^*$ by an edge incident to $\fff{i}{j}$ whose second endpoint is the same as for $e^*$ unless $e^*$ was deleted by some modification for $F_1$.
\end{itemize}
For $q=1$, we formally set $X_f=\{f\}$ and  $f_1=f$, that is, we do not perform any modification.   
The construction is shown in \Cref{fig:F2}.
Notice that the vertices~$\fff{i}{1},\fff{i}{2},\ldots,\fff{i}{q}$ can be embedded in the face~$\ff{i}$ of $G$ such that the resulting graph $H^*$ is plane.
For each edge~$e^* \in E(H^*)$, there is an edge~$e\in E(G^*)$ such that $e^*$ was constructed from the edge that is dual to~$e$ in~$G^*$. Slightly abusing notation, we do not distinguish the edges of~$H^*$ and~$G^*$. In particular, we say that~$e^*$ is dual to~$e$. 


\medskip
Our next aim is to assign labels to the edges of $H^*$ from $\mathbb{Z}_2^d$ for some appropriate $d$. For this, we greedily pick a set $R$ of \emph{representatives} from $S\cup T$ for each $f\in F'$. From each $f\in F_1$, we select two terminals from $S$ and $T$, respectively, that are on the frontier of the face $f$ of $G$. For each $f\in F_2$, we pick one terminal from the frontier of the face of $f$. 
Then, we construct an arbitrary inclusion minimal tree $Q$ in $G$ that spans $R$. This can be done in linear time using standard tools (see, e.g., \cite{CormenLRS09}). We select an arbitrary vertex $u\in R\cap S$ and set $d=|R|-1$. Observe that $|R|\leq 2|F_1|+|F_2|$ and $d\leq 2|F_1|+|F_2|-1$.
Denote by $v_1,\ldots,v_d$ the vertices of $L\setminus \{u\}$ and let $Q_i$ be the $u$-$v_i$-path in~$Q$ for each~$i\in[d]$.
We define~$g\colon E(G)\rightarrow \mathbb{Z}_2^d$ by setting $g(e)=(\delta_1,\ldots,\delta_d)^\intercal$ where for each $i\in[d]$, 
$\delta_i=
\begin{cases}
1&\mbox{if }e\in E(Q_i),\\
0&\mbox{if }e\notin E(Q_i).
\end{cases}$

Moreover, let~$g^*\colon E(H^*)\rightarrow \mathbb{Z}_2^d$ be defined by setting $g^*(e^*)=g(e)$ for each~$e^*\in E(H^*)$ that is dual to~$e\in E(G)$ and let~$c=(c_1,\ldots,c_d)^\intercal\in \mathbb{Z}_2^d$ where 
$$c_i=
\begin{cases}
0&\mbox{if }v_i\in S,\\
1&\mbox{if }v_i\in T
\end{cases}
\text{ for }i\in[d].$$

We show the following claim.

\begin{claim}\label{cl:main}
 The graph~$G^*$ contains a cycle $C^*$ of length at most $k$ such that the vertices of~$S$ and the vertices of~$T$ are in distinct faces of $C^*$ if and only if the instance $(H^*,g^*,c,\{X_f\mid f\in F_2\})$ of \probSC has a solution and the length of a solution cycle is at most $k$.
  \end{claim}

\begin{proof}[Proof of \Cref{cl:main}]
Suppose that $C^*$ is a cycle in $G^*$ of length at most $k$ such that~$S$ and~$T$ are in distinct faces of $C^*$.
Note that $C^*$ corresponds to a cycle $\hat{C}^*$ in $H^*$ by replacing each vertex corresponding to a face~$\ff{i}\in F_2$ by a certain vertex $\fff{i}{j}$ obtained by splitting $\ff{i}$. 

Formally, let $\ff{i} \in V(C^*) \cap F_2$. Since~$\ff{i} \in F_2$, the frontier cycle~$C'$ of $\ff{i}$ contains either only vertices of $S$ or only vertices of $T$. We assume without loss of generality that it contains only vertices from~$S$. Since the vertices of $S$ and the vertices of $T$ are in distinct faces of~$C^*$, we have that the vertices of $S'$ are in the same face of $C^*$. By \Cref{lem:nonsepset}, $C'$ crosses $C^*$ in exactly two edges $e_1$ and $e_2$ such that the vertices of $S'$ are in the same connected component of $C'-\{e_1,e_2\}$. This implies that $e_1$ and $e_2$ are edges of some path $P_j$ in  $C'$ with end-vertices in $S'$ whose internal vertices are not in $S'$. Recall that $H^*$ has a vertex $\fff{i}{j}$ that is incident to the same edges as the vertex~$\ff{i} \in V(C^*)$. Thus, we can replace $\ff{i}$ in~$C^*$ by~$\fff{i}{j}$. 

 To argue that every edge of $\hat{C}^*$ is an edge of $H^*$, note that each edge $e\in E(G^*)$ that is not an edge of $H^*$ is incident to some vertex~$f$ such that the face~$f \in F_1$. Suppose that there is $f \in V(C^*)$ such that~$f\in F_1$. Let~$C'$ be the frontier cycle of~$f$. 
 Note that~$S'=S\cap V(C') \neq \emptyset$ and $T'=T\cap V(C') \neq\emptyset$. Since the vertices of $S'$ and the vertices of $T'$ are in distinct faces of $C^*$, by \Cref{lem:sepsets}, $C'$ crosses $C^*$ in two edges $e_1$ and $e_2$ such that (i) the vertices of $S'$ are in the same connected component of $C'-\{e_1,e_2\}$, (ii)  the vertices of $T'$ are in the same connected component of $C'-\{e_1,e_2\}$, and (iii) the vertices of $S'$ and the vertices of $T'$ are in distinct connected components of $C'-\{e_1,e_2\}$.  For the inclusion-minimal disjoint paths $P_1$ and $P_2$ in~$C'$ such that $S'\subseteq V(P_1)$ and $T'\subseteq V(P_2)$, we have that~$e_1,e_2\notin E(P_1)\cup E(P_2)$. Thus, the edges $e_1^*$ and $e_2^*$ that are incident to~$\ff{i}$ in~$C^*$ are edges of $H^*$.  This concludes the proof that $\hat{C}^*$ is a cycle in $H^*$.

Clearly, $\hat{C}^*$ has the same length as $C^*$. Also, the construction implies that $\hat{C}^*$ contains at most one vertex of $X_i$ for every $i\in[p]$. Let $i\in[d]$ and consider~$c_i$ and the $i$-th coordinate~$g^*(e^*)[i]$ of~$g^*(e^*)$ for some~$e^*\in E(\hat{C}^*)$. Suppose that for the representative terminal $v_i\in R$, it holds that $v_i\in S$. Since~$u\in S$ and~$v_i\in S$, we have that~$u$ and~$v_i$ are in the same face of~$C^*$ and
the number of crosses of the path~$Q_i$ and~$C^*$ is even by \Cref{obs:cycle-sep}. The construction of $\hat{C}^*$ implies that the number of crosses of~$Q_i$ and~$\hat{C}^*$ is even as well. Then,
$\sum_{e^*\in E(\hat{C}^*)}g^*(e^*)[i]=0=c_i$. Similarly, if~$v_i\in T$, we obtain that the number of crosses of $Q_i$ and $\hat{C}^*$ is odd and  $\sum_{e^*\in E(\hat{C}^*)}g^*(e^*)[i]=1=c_i$. Thus,
$\sum_{e^*\in E(\hat{C}^*)}g^*(e^*)=c$. We conclude that~$\hat{C}^*$ is a cycle satisfying the constraints of \probSC. Thus, the instance~$(H^*,g^*,c,\{X_f\mid f\in F_2\})$ of \probSC has a solution cycle whose length is at most~$k$.

For the opposite direction, suppose that  $(H^*,g^*,c,\{X_f\mid f\in F_2\})$ of \probSC has a solution cycle $\hat{C}^*$ whose length is at most $k$. Observe that since~$\hat{C}^*$ contains at most one vertex of~$X_f$ for each $f\in F_2$, $\hat{C}^*$ corresponds to the cycle~$C^*$ in $G^*$ obtained by replacing each vertex~$\fff{i}{j}$ constructed for a face~$\ff{i}\in F_2$ by the vertex~$\ff{i}\in F$. Trivially, $C^*$ has the same length as $\hat{C}^*$. We also have that $\sum_{e^*\in E(C^*)}g^*(e^*)=\sum_{e^*\in E(\hat{C}^*)}g^*(e^*)=c$.
Consider any dimension~$i\in[d]$. If~$v_i\in S$, then $\sum_{e^*\in E(C^*)}g^*(e^*)[i]=c_i=0$. Hence, $u$ and $v_i$ are in the same face of $C^*$ by \Cref{obs:cycle-sep}. If $v_i\in T$, then  $\sum_{e^*\in E(C^*)}g^*(e^*)[i]=c_i=1$ and  $u$ and $v_i$ are in distinct  faces of $C^*$. This proves that the vertices of $S'=S\cap R$ and $T'=T\cap R$ are in distinct faces of $C^*$.

We claim that the vertices of~$S$ and the vertices of $T$ are in distinct faces of $C^*$. To show this, consider a vertex~$v\in (S\cup T)\setminus R$.  By symmetry, we assume without loss of generality that $v\in S$. Since the vertices of~$S'$ and the vertices of~$T'$ are in distinct faces of~$C^*$, it is sufficient to show that there is some~$s\in S'$ such that $s$ and $v$ are in the same face of~$C^*$. Let~$f \in F'$ be a face that covers~$v$. 

Suppose first that $f\in F_1$. Then, $R$ contains two representative vertices~$s\in S$ and~$t\in T$ from the cycle~$C'$ forming the frontier of $f$. Since~$s$ and $t$ are in distinct faces of $C^*$, by \Cref{lem:sepsets}, $f\in V(C^*)$ and  $C'$ crosses $C^*$ in two edges $e_1$ and $e_2$ such that $s$ and $t$ are in distinct connected components of $C'-\{e_1,e_2\}$. By the construction of $H^*$, there is a path $P_1$ in $C'$ such that all vertices in~$S$ incident to~$f$ are contained in~$P_1$ and  $e^*\notin E(H^*)$ for each edge~$e$ in~$P_1$. Hence, $P_1$ is a path in $C'-\{e_1,e_2\}$ because~$e_1^*,e_2^*\in E(H^*)$. By \Cref{lem:sepsets}, the vertex $t$ and the vertices of $S$ incident to~$f$ are in distinct faces of $C^*$. In particular,~$v$ is a different face of~$C^*$ than~$t$ and it is therefore in the same face of $C^*$ as~$s$. 

Assume now that $\ff{i} \in F_2$. Then, $R$ contains a representative~$s\in S$ such that~$s$ is in the frontier of $\ff{i}$. If~$\ff{i} \notin V(C^*)$, then~$s$ and~$v$ are in the same face of~$C^*$ by \Cref{lem:nonsepset}.
Suppose that $\ff{i} \in V(C^*)$. Then, the frontier cycle~$C'$ of~$\ff{i}$ crosses $C^*$ in two edges $e_1$ and $e_2$.
We remind that $C^*$ is constructed from the cycle~$\hat{C}^*$. By the construction of $H^*$, there is a path $P_j$ whose internal vertices do not belong to~$S$ and the edges $e_1$ and $e_2$ are contained in~$P_j$.
Thus, $s$ and $v$ are in the same connected component of $C'-\{e_1,e_2\}$. By \Cref{lem:nonsepset}, $s$ and $v$ are in the same face of $C^*$. This concludes the proof of the claim. 
\end{proof}

By \Cref{cl:main}, solving \probCutUncut for $(G,S,T,k)$ is equivalent to solving  \probSC for $(H^*,g^*,c,\{X_f\mid f\in F_2\})$. For this, we use the algorithm from \Cref{cor:group-labeled-cycle}. 

To evaluate the running time, observe that a face cover $F'$ (if it exists) of size at most $r$ can be constructed in $2^{\Oh(\sqrt{r})}\cdot n^{\Oh(1)}$ time. Given such a cover, the graph $H^*$ together with the sets $X_f$ for $f\in F_2$ can be constructed in polynomial time. 
Because $d\leq 2|F_1|+|F_2|-1\leq 2r$, the lableling $g^*$ and $c\in \mathbb{Z}_2^d$ also can be constructed in polynomial time. Finally, because $d\leq 2|F_1|+|F_2|-1$ and $p=|\{X_i\mid f\in F_2\}|=|F_2|$, we have that $p+d\leq 2r-1$ and the algorithm from \Cref{cor:group-labeled-cycle} runs in $4^{r}\cdot n^{\Oh(1)}$ time. We conclude that the overall running time is  $4^{r+\Oh(\sqrt{r})}\cdot n^{\Oh(1)}$. This completes the proof.

\medskip
The above algorithm for the parameterization by the size of a face cover  uses a planar embedding of $G$ because we have to find a set of faces covering the terminals. However, if we parameterize \probCutUncut by $\ell=|S|+|T|$, then an embedding is not needed and we can use a simplified variant of the algorithm. 
Given an instance $(G,S,T,k)$ of \probCutUncut where $G$ is a planar graph, we use the classical algorithm of Hopcroft and Tarjan~\cite{HopcroftT74} to find a plane embedding of $G$. Then we use the variant of the algorithm where we do not modify $G^*$, that is, we set $H^*=G^*$, and where we assume that all the terminals are representatives, that is, we set $R=S\cup T$. The labeling~$g^*\colon E(H^*)\rightarrow \mathbb{Z}_2^d$ and $c$ are defined in the same way as in the algorithm for the parameterization by the size of a face cover.  By \Cref{obs:cycle-sep}, solving \probCutUncut for $(G,S,T,k)$ is equivalent to solving  \probSC for $(H^*,g^*,c,\emptyset)$. Since~$d=|R|-1=|S|+|T|-1$, we conclude that we can solve the problem in $2^{|S|+|T|}\cdot n^{\Oh(1)}$ time by the algorithm from \Cref{cor:group-labeled-cycle}. 
\end{proof}

\section{Hardness}
\label{sec:hardness}

It is known that \probCutUncut{} is NP-complete~\cite{GrayKLS12} in planar graphs and that it is \classNP-complete in general graphs even if~$|S| = 2$~\cite{HofPW09}.
We strengthen the latter result by showing that \probCutUncut{} remains \classW{1}-hard parameterized by~$|T|$ even if~$|S|=1$ by providing a polynomial-time reduction from \name{Regular Multicolored Clique} parameterized by solution size~$k$---a variant of \name{Multicolored Clique} where each vertex has the same degree~$d$---such that~$|T| = k$.
This problem is known to be \classW{1}-hard and assuming ETH, it cannot be solved in~$f(k) \cdot n^{o(k)}$ time.

\begin{proposition}
 \probCutUncut{} is \classW{1}-hard when parameterized by~$|T|$ even if~$|S|=1$. Moreover, this restricted version cannot be solved in~$f(|T|) \cdot n^{o(k)}$ time unless the ETH breaks.
\end{proposition}

\begin{proof}
Let~$(G,k)$ be an instance of \name{Regular Multicolored Clique}.
Let~$V(G) = \{v_1,v_2,\ldots,v_n\}$, let~$(V_1,V_2,\ldots,V_{k})$ be the~$k$ partition of~$V(G)$, and let~$d$ be the degree of each vertex in~$G$.
We construct an equivalent instance~$(H,S,T,\ell)$ of~\probCutUncut{} as follows.
The graph~$H$ contains~$G$ as an induced subgraph.
Moreover, it contains new vertices~$s,t_1,t_2,\ldots,t_k$ and~$v_i^j$ for each combination of~$i \in [n]$ and~$j \in [n+2m]$.
The vertex~$s$ is connected to~$v_i^j$ for each combination of~$i \in [n]$ and~$j \in [n+m]$.
Moreover, for each~$i\in [n]$, each vertex~$v_i$ is connected to each vertices~$v_i^j$ with~$j \in [n+2m]$.
Finally, for each~$i \in [k]$, the vertex~$t_i$ is connected to all vertices in~$V_i$.
We set~$S=\{s\}$, $T = \{t_1,t_2,\ldots,t_k\}$ and~$\ell=n-k+k(n+2m)+k(d-k+1)$.
This concludes the construction.
Note that~$|T|=k$ and that the reduction takes polynomial time.

It only remains to show that the two instances are equivalent.
To this end, assume that there is a multicolored clique of size~$k$ in~$G$.
For the sake of notational ease, let us assume that the clique consists of vertices~$C = \{v_1,v_2,\ldots,v_k\}$ and let~$v_i \in V_i$ for each~$i \in [k]$.
We delete~$\ell$ edges as follows.
For each~$V_i$, we delete all edges between vertices in~$V_i \setminus \{v_i\}$ and~$t_i$.
We also delete all edges between~$v_i$ and~$v_i^j$ for each combination of~$i \in [k]$ and~$j \in [n+2m]$.
Finally, we delete all edges in~$G$ that have exactly one endpoint in~$C$.
Overall, we have removed~$n-k + k(n+2m) + k(d-k+1)$ edges as each vertex in~$C$ has~$d$ incident edges and~$(k-1)$ of those have the other endpoint also in~$C$.
Note that the resulting graph contains two connected components, one containing all vertices in~$T \cup C$ and the other containing all other vertices.
Thus the resulting instance of \probCutUncut{} is also a yes-instance.

Coversely, assume that there is a set of~$\ell$ edges whose removal disconnects~$S$ from~$T$ while maintaining connectivity between the vertices in~$T$.
Note that in order to maintain connectivity between the vertices in~$T$, at least one edge between~$t_i$ and some vertex in~$V_i$ has to be maintained.
Again, we will assume for notational ease that the edge between~$t_i$ and~$v_i$ remains in the graph.
Note that in order to disconnect~$t_i$ from~$s$, now vertex~$v_i$ needs to be disconnected from~$s$ and thus at least one of the edges~$sv_i^j$ or~$v_i^jv_i$ has to be deleted for each~$i\in [k]$ and each~$j \in [n+2m]$.
Moreover, if there are at least~$k+1$ vertices in~$V$ that do not belong to the connected component of~$s$ in the solution, then we need to delete at least~$k+1 \cdot (n + 2m) > \ell$ edges.
The inequality holds since~$k(d-k+1) \leq kd < nd = 2m$.
We show that the set~$C = \{v_i \mid i \in [k]\}$ of vertices forms a multicolored clique.
First, by construction each vertex in~$C$ is connected to a different vertex~$t_i$ and~$C$ is therefore a multicolored set of vertices.
Second, note that we need to remove all edges between vertices in the connected component of~$s$ and the selected vertices~$v_i$ with~$i \in [k]$.
Since we have already removed~$n-k + k \cdot (n+2m)$ edges, we can remove at most~$k(d-k+1)$ edges.
Moreover, since each vertex has degree~$d$ in~$G$, we have to remove~$d-x$ edges incident to each vertex~$v_i \in C$ where~$x$ is the number of neighbors of~$v_i$ in~$C$.
Note that since there are only~$k$ vertices in~$C$, the minimum number of edges to remove is~$d-k+1$ and this bound can only be achived if~$v_i$ is connected to all other vertices in~$C$.
Since we assumed that there is a solution which removes only~$\ell$ edges, we can infer that each pair of vertices in~$C$ is pairwise adjacent, that is,~$C$ is a multicolored clique.
This concludes the proof.
\end{proof}




\section{Applications}
\label{sec:applications}

\newcommand{\FA}{\ensuremath{\mathcal{F}_A}}
\newcommand{\FB}{\ensuremath{\mathcal{F}_B}}

\newcommand{\subin}[1]{\ensuremath{#1_\textnormal{in}}}
\newcommand{\subout}[1]{\ensuremath{#1_\textnormal{out}}}

In this section, we show how to use our result to generalize two known results from the literature.
We also complement the \classFPT-time algorithms with \classNP-hardness results.
First, we show that a generalization of \probND{} with more than one bridge can be solved in (randomized) \classFPT-time when parameterized by the number of bridges.
The problem \probGND{} is defined as follows.\footnote{We mention in passing that a problem where one only wants to remove a certain number of edges to ensure that each remaining~$s$-$t$-path contains at least one edge in~$B$ has been considered before. However, it has been noted by Cintron--Arias et al.~\cite{cintron-arias2001networkdiversion} that this problem reduces to
  (weighted)
  the case with only a single bridge. We therefore believe that it makes more sense to demand \emph{every} edge in~$B$ to become a bridge, that is, to require that~$B$ is contained in the minimal cut.}

\subsection{Generalized Network Diversion}
\label{sec:appl-network-diversion}

\begin{XProblem}
  {\probGND}
  \Input & A graph $G$, two vertices $s$ and $t$, a set $B$ of edges of $G$, and an integer $k\geq 0$. \\
  \Prob & Decide whether there exists a minimal~$s$-$t$-cut~$(U,W)$ of $G$ with $\left| \cut(U) \right| \leq k$ and $B\subseteq \cut(U)$.
\end{XProblem}

As a simple corollary of \Cref{thm:face-cover-param}, we get also an
\classFPT-time algorithm for \probGND.  We show that this problem is
indeed \classNP-complete in the appendix of this article, in
\Cref{prop:gnd-np-c}.

 \begin{corollary}\label{cor:ND}
\probGND can be solved in $8^{|B|}\cdot n^{\Oh(1)}$ time on planar graphs by a randomized algorithm with one-sided error.
\end{corollary}

\begin{proof}
Let~$(G,s,t,B,k)$ be an instance of \probGND.  We consider all possible sets of terminals $S$ and $T$ by guessing which endpoint of each edge in~$B$ is on the same side of a solution cut as~$s$.  Initially, $S:=\{s\}$ and $T:=\{t\}$. Then we have $2^{|B|}$ possibilities to include one endpoint of each edge~$e\in B$ in~$S$ and the other in~$T$. For each choice of $S$ and $T$, we run the algorithm from \Cref{thm:face-cover-param}. Since~$|S|+|T|\leq 2|B|+2$, the overall running time is $8^{|B|}\cdot n^{\Oh(1)}$.
\end{proof}

Complementing the algorithm above, we show that the \classFPT running time can probably not be improved to a polynomial running time in the appendix of this article, in \Cref{prop:gnd-np-c}.

\subsection{Location Constrained Shortest Path}
\label{sec:appl-locat-constr}

Our second application regards the problem \ProblemName{Location Constrained Shortest Path} studied by Duan and Xu~\cite{duan2014connectivitypreserving}.

\begin{XProblem}
  {Location Constrained Shortest Path}
  \Input & A plane graph $G$, two vertices $s$ and $t$ on the outer face and an interior face $F$\\
  \Prob & Find a shortest interior $s$-$t$ path \emph{below} $F$.
\end{XProblem}

An \emph{interior path} is a path in which only the endpoints can be on the outer face.
Given an interior~$s$-$t$-path $P$, we can extend it to a simple cycle~$C$
by appending the path on the outer face from $t$ to $s$ which is
on the lower side of the graph. This is a simple cycle since $P$ does
not use any exterior vertices except~$s$ and~$t$.
A face $F$ is now said to be \emph{below $P$} if~$F$ is inside~$C$.
Duan and Xu~\cite{duan2014connectivitypreserving} give a polynomial time
algorithm for \ProblemName{Location Constrained Shortest Path}.

We consider a generalization of  \ProblemName{Location Constrained Shortest Path} which we call \ProblemName{Generalized Location Constrained Shortest Path}.
Given a plane graph~$G$ and two sets $\mathcal{F}_A = F^A_1, F^A_2, \dots, F^A_{p_A}$ and $\mathcal{F}_B = F^B_1, F^B_2, \dots, F^B_{p_B}$ of faces of~$G$, the task is to find a shortest interior~$s$-$t$-path~$P$ such that all faces in~$\mathcal{F}_A$ are \emph{above}~$P$ and all faces in~$\mathcal{F}_B$ are \emph{below}~$P$.

\begin{XProblem}
  {Generalized Location Constrained Shortest Path}
  \Input & A plane graph $G$, two vertices $s$ and $t$ on the outer
  face and interior faces $\FA$ and $\FB$\\
  \Prob & Find a shortest interior $s$-$t$ path $P$ such that $\FA$ is
  above $P$ and $\FB$ is below $P$.
\end{XProblem}

We generalize the result by Duan and Xu by showing that \ProblemName{Generalized Location Constrained Shortest Path} is fixed-parameter tractible when parameterized by $|\mathcal{F}_A| + |\mathcal{F}_B|$.

\begin{proposition}
  \ProblemName{Generalized Location Constrained Shortest Path} can be solved in randomized~$2^{|\FA \cup \FB|} \cdot n^{\Oh(1)}$ time with one-sided error.
\end{proposition}

\begin{proof}
  We show that  \ProblemName{Generalized Location Constrained Shortest Path} is a special case of \probCutUncut.
  Let $(G, s, t, \FA, \FB)$ be the input to \ProblemName{Generalized Location Constrained Shortest Path}, with $\mathcal{O}$ the vertices on the outer face except $s$ and $t$.  Notice that in the graph $G - \mathcal{O}$, $s$ and $t$ must be in the same connected component, otherwise we have a trivial no-instance.
  Let $G'$ be the graph $G - \mathcal{O} + \{s,t\}$, i.e., the graph where we add one additional edge between~$s$ and~$t$ where we draw this edge in the embedding above the rest of the graph.
  We call the new (non-outer) face~$A$ and the new outer face~$B$.
   Now, let $S' = \FA \cup \{A\}$ and $T' = \FB \cup \{B\}$ and let~$G^*$ be
  the dual graph of $G'$.  Let $S$ be the vertices corresponding to the
  faces of $S'$ and $T$ the vertices corresponding to the vertices of
  $T'$. We show next that $(G^*, S,T,k)$ is a yes-instance for \probCutUncut if and only if $(G, s, t, \FA, \FB)$ has an interior $s$-$t$-path of length $k-1$.

  Since the algorithm for \probCutUncut finds a minimal cut in $G^*$ of
  size at most $k$, this corresponds to a simple cycle of length at most $k$ in
  $G'$.
  Note that the two faces $A$ and $B$ are incident to each other,
  hence the edge between the corresponding two vertices in $G^*$ must be part of the cut.
  This means that the new edge~$\{s,t\}$ is part of the cycle, which means that the rest of the cycle is an interior $s$-$t$-path of length at most~$k-1$.
\end{proof}

We defer the
\classNP-completeness
of
\ProblemName{Generalized Location Constrained Shortest Path}
to the appendix. 

\bigskip%
\noindent%
These two problems, \probGND{} and \ProblemName{Generalized Location
  Constrained Shortest Path} are just two applications that is solved
directly by our \probCutUncut{} algorithm from
\Cref{thm:face-cover-param}.  We assume there are many more
applications.
%


\section{Conclusion}
\label{sec:conclusion}

In this paper, we showed that  \probCutUncut is \classFPT on planar graphs parameterized by the number of terminals.
We also prove a more general result that the problem remains FPT parameterized by the minimum number of faces required to cover the terminals. Our result implies a polynomial time algorithm solving \probND on planar graphs. 
We complement this result by showing that \probCutUncut{} parameterized by the number of terminals ($|S|+|T|$) is W[1]-hard in general graphs even when~$|S|=1$.

%

First, let us remark that the algorithm in \Cref{thm:face-cover-param} for  \probCutUncut{} parameterized by the number of faces is given for plane graphs because the minimum number of faces covering the terminals depends on the embedding. However, the standard techniques based on SPQR trees~\cite{GutwengerM00,HopcroftT73} can be used to show that it is \classFPT{} to decide,
given a planar graph $G$, a set of vertices $X$, and an integer $r\geq 1$, whether $G$ admits a plane embedding such that $X$ can be covered by at most $r$ faces. Thus, the result can be extended (with worse running time) to planar graphs admitting embeddings such that the terminals can be covered by at most $r$ faces.

We conclude with a few open problems.
\begin{enumerate}
\item First we repeat the long-standing open question, whether \probND{} is polynomial-time
  solvable in general graphs. Similar question is valid even for a graph embeddable in a torus. 
\item A natural extension of the   \probCutUncut is to extend it to a larger number of sets. Since on general graphs, \textsc{$3$-Way Cut} is \classNP-complete   
\cite{DahlhausJPSY94}, the same holds for  \ProblemName{Three-Sets Cut-Uncut} even when all sets are of size one. 
However, for planar graphs, \textsc{$k$-Way Cut} is solvable in polynomial time for fixed $k$ \cite{DahlhausJPSY94,KleinM12,pandey2022planarmultiway}. 
As a very concrete open question, we ask whether 
  \ProblemName{Three-Sets Cut-Uncut} is 
  solvable in polynomial time on planar graphs when two sets are of size one and one set is of size two.
\item Our algorithm is randomized and works only on unweighted graphs;
  can we get rid of either of these restrictions? 
\end{enumerate}



\bibliographystyle{siam}

\bibliography{bibliography}

\begin{thebibliography}{10}

\bibitem{agnarsson2006graph}
{\sc G.~Agnarsson and R.~Greenlaw}, {\em Graph theory: {Modeling},
  applications, and algorithms}, Prentice-Hall, Inc., 2006.

\bibitem{AlberBFKN02}
{\sc J.~Alber, H.~L. Bodlaender, H.~Fernau, T.~Kloks, and R.~Niedermeier}, {\em
  Fixed parameter algorithms for dominating set and related problems on planar
  graphs}, Algorithmica, 33 (2002), pp.~461--493.

\bibitem{bern1990fasterexact}
{\sc M.~Bern}, {\em Faster exact algorithms for steiner trees in planar
  networks}, Networks, 20 (1990), pp.~109--120.

\bibitem{bezakova2014minimumplanar}
{\sc I.~Bez{\'{a}}kov{\'{a}} and Z.~Langley}, {\em Minimum planar multi-sink
  cuts with connectivity priors}, in Proceedings of the 39th International
  Symposium on Mathematical Foundations of Computer Science ({MFCS}~'14),
  Springer, 2014, pp.~94--105.

\bibitem{DBLP:journals/jcss/BjorklundHKK17}
{\sc A.~Bj{\"{o}}rklund, T.~Husfeldt, P.~Kaski, and M.~Koivisto}, {\em Narrow
  sieves for parameterized paths and packings}, J. Comput. Syst. Sci., 87
  (2017), pp.~119--139.

\bibitem{bjorklund2012shortestcycle}
{\sc A.~Bj{\"o}rklund, T.~Husfeldt, and N.~Taslaman}, {\em Shortest cycle
  through specified elements}, in Proceedings of the Annual {ACM}-{SIAM}
  Symposium on Discrete Algorithms ({SODA}~'12), Society for Industrial and
  Applied Mathematics, 2012, pp.~1747--1753.

\bibitem{ChitnisCHPP16}
{\sc R.~Chitnis, M.~Cygan, M.~Hajiaghayi, M.~Pilipczuk, and M.~Pilipczuk}, {\em
  Designing {FPT} algorithms for cut problems using randomized contractions},
  {SIAM} J. Comput., 45 (2016), pp.~1171--1229.

\bibitem{cintron-arias2001networkdiversion}
{\sc A.~Cintron-Arias, N.~Curet, L.~Denogean, R.~Ellis, C.~Gonzalez,
  S.~Oruganti, and P.~Quillen}, {\em A network diversion vulnerability
  problem}, 2001.

\bibitem{CormenLRS09}
{\sc T.~H. Cormen, C.~E. Leiserson, R.~L. Rivest, and C.~Stein}, {\em
  Introduction to Algorithms}, {MIT} Press, 2009.

\bibitem{cullenbine2013theoreticalcomputational}
{\sc C.~A. Cullenbine, R.~K. Wood, and A.~M. Newman}, {\em Theoretical and
  computational advances for network diversion}, Networks, 62 (2013),
  pp.~225--242.

\bibitem{curet2001networkdiversion}
{\sc N.~D. Curet}, {\em The network diversion problem}, Military Operations
  Research, 6 (2001), pp.~35--44.

\bibitem{CyganKLPPSW21}
{\sc M.~Cygan, P.~Komosa, D.~Lokshtanov, M.~Pilipczuk, M.~Pilipczuk,
  S.~Saurabh, and M.~Wahlstr{\"{o}}m}, {\em Randomized contractions meet lean
  decompositions}, {ACM} Trans. Algorithms, 17 (2021), pp.~6:1--6:30.

\bibitem{CyganPPW14a}
{\sc M.~Cygan, M.~Pilipczuk, M.~Pilipczuk, and J.~O. Wojtaszczyk}, {\em Solving
  the 2-disjoint connected subgraphs problem faster than $2^n$}, Algorithmica,
  70 (2014), pp.~195--207.

\bibitem{DahlhausJPSY94}
{\sc E.~Dahlhaus, D.~S. Johnson, C.~H. Papadimitriou, P.~D. Seymour, and
  M.~Yannakakis}, {\em The complexity of multiterminal cuts}, {SIAM} J.
  Comput., 23 (1994), pp.~864--894.

\bibitem{derigs1985efficient}
{\sc U.~Derigs}, {\em An efficient {Dijkstra}-like labeling method for
  computing shortest odd/even paths}, Information Processing Letters, 21
  (1985), pp.~253--258.

\bibitem{Diestel12}
{\sc R.~Diestel}, {\em Graph Theory}, Springer, 2012.

\bibitem{duan2015ddosattack}
{\sc Q.~Duan, H.~Jafarian, E.~Al-Shaer, and J.~Xu}, {\em On {DDoS} attack
  related minimum cut problems}, 2015.

\bibitem{duan2014connectivitypreserving}
{\sc Q.~Duan and J.~Xu}, {\em On the connectivity preserving minimum cut
  problem}, Journal of Computer and System Sciences, 80 (2014), pp.~837--848.

\bibitem{erickson1987send}
{\sc R.~E. Erickson, C.~L. Monma, and A.~F.~J. Veinott}, {\em Send-and-{Split}
  {Method} for {Minimum}-{Concave}-{Cost} {Network} {Flows}}, Mathematics of
  Operations Research, 12 (1987), pp.~634--664.
\newblock Publisher: INFORMS.

\bibitem{erkenozgur2002branchandboundalgorithm}
{\sc O.~Erken}, {\em A {Branch}-and-bound {Algorithm} for the {Network}
  {Diversion} {Problem}}, Master's thesis ADA411769, Naval Postgraduate School
  Monterey Ca, Dec. 2002.
\newblock Section: Technical Reports.

\bibitem{filtser2019facecover}
{\sc A.~Filtser}, {\em A face cover perspective to $\ell_1$ embeddings of
  planar graphs}, in Proceedings of the 2020 {ACM}-{SIAM} {Symposium} on
  {Discrete} {Algorithms} ({SODA}), Proceedings, Society for Industrial and
  Applied Mathematics, Dec. 2019, pp.~1945--1954.

\bibitem{DBLP:conf/soda/FominGKSS23}
{\sc F.~V. Fomin, P.~A. Golovach, T.~Korhonen, K.~Simonov, and G.~Stamoulis},
  {\em Fixed-parameter tractability of maximum colored path and beyond}, in
  Proceedings of the 2023 {ACM-SIAM} Symposium on Discrete Algorithms, {SODA}
  2023, Florence, Italy, January 22-25, 2023, N.~Bansal and V.~Nagarajan, eds.,
  {SIAM}, 2023, pp.~3700--3712.

\bibitem{garey1979computers}
{\sc M.~R. Garey and D.~S. Johnson}, {\em Computers and {Intractability}, {A}
  {Guide} to the {Theory} of {NP}-{Completeness}}, W.H. Freeman and Company,
  New York, 1979.

\bibitem{GrayKLS12}
{\sc C.~Gray, F.~Kammer, M.~L{\"{o}}ffler, and R.~I. Silveira}, {\em Removing
  local extrema from imprecise terrains}, Comput. Geom., 45 (2012),
  pp.~334--349.

\bibitem{grotschel1981weakly}
{\sc M.~Gr{\"o}tschel and W.~R. Pulleyblank}, {\em Weakly bipartite graphs and
  the {Max}-cut problem}, Operations Research Letters, 1 (1981), pp.~23--27.

\bibitem{GutwengerM00}
{\sc C.~Gutwenger and P.~Mutzel}, {\em A linear time implementation of
  spqr-trees}, in Graph Drawing, 8th International Symposium, {GD} 2000,
  Colonial Williamsburg, VA, USA, September 20-23, 2000, Proceedings, J.~Marks,
  ed., vol.~1984 of Lecture Notes in Computer Science, Springer, 2000,
  pp.~77--90.

\bibitem{HopcroftT73}
{\sc J.~E. Hopcroft and R.~E. Tarjan}, {\em Dividing a graph into triconnected
  components}, {SIAM} J. Comput., 2 (1973), pp.~135--158.

\bibitem{HopcroftT74}
\leavevmode\vrule height 2pt depth -1.6pt width 23pt, {\em Efficient planarity
  testing}, Journal of the {ACM}, 21 (1974), pp.~549--568.

\bibitem{iwata2022findingshortest}
{\sc Y.~Iwata and Y.~Yamaguchi}, {\em Finding a {Shortest} {Non}-{Zero} {Path}
  in {Group}-{Labeled} {Graphs}}, Combinatorica, 42 (2022), pp.~1253--1282.

\bibitem{kallemyn2015modelingnetwork}
{\sc B.~S. Kallemyn}, {\em Modeling {Network} {Interdiction} {Tasks}}, doctoral
  thesis, Department Of The Air Force Air University Air Force Institute Of
  Technology, Wright-Patterson Air Force Base, Ohio, 2015.

\bibitem{KammerT12}
{\sc F.~Kammer and T.~Tholey}, {\em The complexity of minimum convex coloring},
  Discret. Appl. Math., 160 (2012), pp.~810--833.

\bibitem{kisfaludibak2020nearly}
{\sc S.~Kisfaludi-Bak, J.~Nederlof, and E.~J.~v. Leeuwen}, {\em Nearly
  {ETH}-tight {Algorithms} for {Planar} {Steiner} {Tree} with {Terminals} on
  {Few} {Faces}}, ACM Transactions on Algorithms, 16 (2020), pp.~1--30.

\bibitem{KleinM12}
{\sc P.~N. Klein and D.~Marx}, {\em Solving planar $k$-terminal cut in
  ${O}(n^{c\sqrt{k}})$ time}, in Proceedings of the 39th International
  Colloquium on Automata, Languages, and Programming (ICALP), vol.~7391 of
  Lecture Notes in Computer Science, Springer, 2012, pp.~569--580.

\bibitem{kobayashi2017finding}
{\sc Y.~Kobayashi and S.~Toyooka}, {\em Finding a {Shortest} {Non}-zero {Path}
  in {Group}-{Labeled} {Graphs} via {Permanent} {Computation}}, Algorithmica,
  77 (2017), pp.~1128--1142.

\bibitem{DBLP:conf/icalp/Koutis08}
{\sc I.~Koutis}, {\em Faster algebraic algorithms for path and packing
  problems}, in Automata, Languages and Programming, 35th International
  Colloquium, {ICALP} 2008, Reykjavik, Iceland, July 7-11, 2008, Proceedings,
  Part {I:} Tack {A:} Algorithms, Automata, Complexity, and Games, 2008,
  pp.~575--586.

\bibitem{krauthgamer2019flow}
{\sc R.~Krauthgamer, J.~R. Lee, and H.~I. Rika}, {\em Flow-{Cut} {Gaps} and
  {Face} {Covers} in {Planar} {Graphs}}, in Proceedings of the 2019 {Annual}
  {ACM}-{SIAM} {Symposium} on {Discrete} {Algorithms} ({SODA}), Proceedings,
  Society for Industrial and Applied Mathematics, Jan. 2019, pp.~525--534.

\bibitem{lee2019combinatorial}
{\sc C.~Lee, D.~Cho, and S.~Park}, {\em A combinatorial benders decomposition
  algorithm for the directed multiflow network diversion problem}, Military
  Operations Research, 24 (2019), pp.~23--40.

\bibitem{Marx12}
{\sc D.~Marx}, {\em A tight lower bound for planar multiway cut with fixed
  number of terminals}, in Proceedings of the 39th International Colloquium on
  Automata, Languages, and Programming (ICALP), vol.~7391 of Lecture Notes in
  Computer Science, Springer, 2012, pp.~677--688.

\bibitem{pandey2022planarmultiway}
{\sc S.~Pandey and E.~J. van Leeuwen}, {\em Planar {Multiway} {Cut} with
  {Terminals} on {Few} {Faces}}, in Proceedings of the 2022 {Annual}
  {ACM}-{SIAM} {Symposium} on {Discrete} {Algorithms} ({SODA}), Proceedings,
  Society for Industrial and Applied Mathematics, Jan. 2022, pp.~2032--2063.

\bibitem{PaulusmaR11}
{\sc D.~Paulusma and J.~M.~M. van Rooij}, {\em On partitioning a graph into two
  connected subgraphs}, Theor. Comput. Sci., 412 (2011), pp.~6761--6769.

\bibitem{001RS16}
{\sc A.~Rai, M.~S. Ramanujan, and S.~Saurabh}, {\em A parameterized algorithm
  for mixed-cut}, in Proceedings of the 12th Latin American Symposium (LATIN),
  vol.~9644 of Lecture Notes in Computer Science, Springer, 2016, pp.~672--685.

\bibitem{Schwartz80}
{\sc J.~T. Schwartz}, {\em Fast probabilistic algorithms for verification of
  polynomial identities}, Journal of the {ACM}, 27 (1980), pp.~701--717.

\bibitem{DBLP:journals/siamcomp/Tarjan72}
{\sc R.~E. Tarjan}, {\em Depth-first search and linear graph algorithms},
  {SIAM} J. Comput., 1 (1972), pp.~146--160.

\bibitem{DBLP:conf/wg/TelleV13}
{\sc J.~A. Telle and Y.~Villanger}, {\em Connecting terminals and 2-disjoint
  connected subgraphs}, in Graph-Theoretic Concepts in Computer Science - 39th
  International Workshop, {WG} 2013, L{\"{u}}beck, Germany, June 19-21, 2013,
  Revised Papers, vol.~8165 of Lecture Notes in Computer Science, Springer,
  2013, pp.~418--428.

\bibitem{HofPW09}
{\sc P.~van~'t Hof, D.~Paulusma, and G.~J. Woeginger}, {\em Partitioning graphs
  into connected parts}, Theor. Comput. Sci., 410 (2009), pp.~4834--4843.

\bibitem{Williams09}
{\sc R.~Williams}, {\em Finding paths of length $k$ in ${O}^*(2^k)$ time}, Inf.
  Process. Lett., 109 (2009), pp.~315--318.

\bibitem{Zippel79}
{\sc R.~Zippel}, {\em Probabilistic algorithms for sparse polynomials}, in
  Proceedings of the International Symposiumon Symbolic and Algebraic
  Computation ({EUROSAM}~'79), Springer, 1979, pp.~216--226.

\end{thebibliography}

\newcommand{\probstRPP}{\ProblemName{$f_s$-$f_t$-separating Rural Postman Problem}}
\newcommand{\probRPP}{\ProblemName{Rural Postman Problem}}
\newcommand{\probRpp}{\ProblemName{RPP}}
\newcommand{\probHC}{\ProblemName{Hamiltonian Cycle}}

\appendix

\section{Proofs from the applications section}
\label{sec:proofs-from-appl}

%
%
%
%

\subsection{Generalized \probND{} is \classNP-complete}
\label{sec:network-div}

We begin this section by recalling the problem definition:

\begin{XProblem}
  {\probGND}
  \Input & A graph $G$, two vertices $s$ and $t$, a set $B$ of edges of $G$, and an integer $k\geq 0$. \\
  \Prob & Decide whether there exists a minimal~$s$-$t$-cut~$(U,W)$ of $G$ with $\left| \cut(U) \right| \leq k$ and $B\subseteq \cut(U)$.
\end{XProblem}

\begin{proposition}
  \label{prop:gnd-np-c}
  \probGND is \classNP-complete, even when restricted to subcubic planar graphs.
\end{proposition}

First, note that \probGND{} on planar graphs is equivalent to finding a
shortest simple cycle in the dual graph that visits all dual edges $B^*$
of edges in~$B$ where the faces~$s^*$ and~$t^*$ lie in different faces
of the cycle.
We call this problem \probstRPP{}, or simply \probRpp{}.
We show that \probRpp{} is \classNP-hard via a reduction from \probHC{} on cubic planar graphs~\cite{garey1979computers}.
Containment in \classNP{} is trivial, so we address only the hardness.

\begin{proof}[Proof of \Cref{prop:gnd-np-c}]
  Let~$G$ be an instance of \probHC{} where~$G$ is a cubic planar graph.
  We will build an equivalent instance~$(G',f_s,f_t,B',k'=n)$ of  \probRpp{} as follows.
  To build~$G'$, we start with~$G$ and replace every vertex with a \thenetg{} and connect its three neighbors to the three pendant vertices, respectively.
  See \Cref{fig:many-bridges-reduction} for an illustration of a \thenetg{} and an example of the described construction.
  We then pick any one of the three edges in the triangle of the \thenet{} to be contained in $B'$.
  Then we take one arbitrary edge $b \in B'$ and let $f_s$ and $f_t$ be
  the two different faces on each side of $b$ (these faces exist
  regardless of embedding).
  Observe that in any cycle $C$ that traverses $b$, the two faces $f_s$
  and $f_t$ will lie in different faces of $C$.


  Let $C$ be a Hamiltonian cycle in a cubic planar graph.  If $u,v,w$ is
  a subpath of $C$, then we can construct the corresponding subpath in
  $G'$ by going from $u$ to $e_v$ (the chosen edge of the triangle in
  the \thenet{} corresponding to $v$) and then to $w$.  Since $C$ is a
  Hamiltonian cycle, the corresponding $C'$ in $G'$ is clearly a
  shortest (simple) cycle.

  Let $C'$ be a cycle that visits all the edges in $B'$.  Since
  $e_v \in B^*$ is an edge in the triangle of a~\thenet{}, the path has
  to come in through a pendant vertex of the \thenet{} and leave through
  a different pendant vertex of the \thenet{}.  Construct the cycle $C$
  in which for each edge $e_v$ traversed in $C'$, we pick $v$ to $C$.
  \begin{claim}
    $C$ is a Hamiltonian cycle in $G$, i.e., 1.\ $C$ is a closed walk,
    2.\ $C$ spans $V(G)$, and 3.\ $C$ does not have repeated vertices.
  \end{claim}

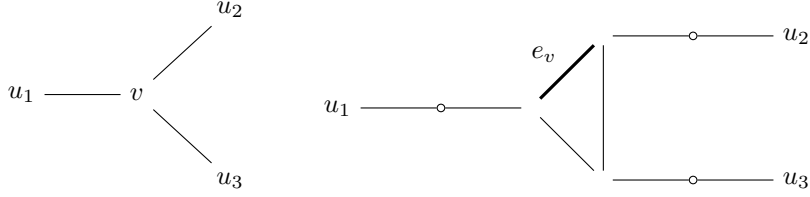
\begin{figure}
  \centering
  \begin{tikzpicture}
    \tikzset{main node/.style={circle,draw,minimum size=1cm,inner sep=0pt}}
    \node (v) [] {$v$};
    \node (p1) [left = 1cm of v] {$u_1$};
    \node (p2) [above right = 1cm of v] {$u_2$};
    \node (p3) [below right = 1cm of v] {$u_3$};

    \path[draw] (v) edge node {} (p1);
    \path[draw] (v) edge node {} (p2);
    \path[draw] (v) edge node {} (p3);
  \end{tikzpicture}
  \qquad
  \begin{tikzpicture}
    \tikzset{main node/.style={circle,draw,minimum size=1cm,inner sep=0pt}}
    \node[] (t1) [] {};
    \node[] (t2) [above right = 1cm of t1] {};
    \node[] (t3) [below right = 1cm of t1] {};

    \node[draw,circle,minimum size=0.1cm, inner sep=0pt] (p1) [left = 1cm of t1]  {};
    \node[draw,circle,minimum size=0.1cm, inner sep=0pt] (p2) [right = 1cm of t2] {};
    \node[draw,circle,minimum size=0.1cm, inner sep=0pt] (p3) [right = 1cm of t3] {};

    \node[] (u1) [left = 1cm of p1] {$u_1$};
    \node[] (u2) [right = 1cm of p2] {$u_2$};
    \node[] (u3) [right = 1cm of p3] {$u_3$};

    \path[draw,very thick] (t1) edge node[above left] {$e_v$} (t2);
    \path[draw] (t2) edge node {} (t3);
    \path[draw] (t1) edge node {} (t3);
    \path[draw] (t1) edge node {} (p1);
    \path[draw] (t2) edge node {} (p2);
    \path[draw] (t3) edge node {} (p3);
    \path[draw] (u1) edge node {} (p1);
    \path[draw] (u2) edge node {} (p2);
    \path[draw] (u3) edge node {} (p3);
  \end{tikzpicture}%
  \caption{Reduction from \probHC{} on
    cubic planar graphs to \probGND. On the left side, a
    vertex~$v$ in the input graph is depicted with its three neighbors. On the
    right side, the \thenetg{} corresponding to~$v$ is shown. Notice that there is a path between any pair of neighbors of $v$ that passes through the edge~$e_v$ in the gadget.}
  \label{fig:many-bridges-reduction}
\end{figure}

\medskip\noindent%
\textit{Proof of claim.}
\begin{enumerate}
\item Suppose that $uv$ is an edge in $C$.  Then this was constructed
  because in $C'$ we visit $e_u$ and then $e_v$.  But $e_u$ and $e_v$
  are then neighboring \thenetgs{}, hence $uv$ is an edge in $G$.
\item Every vertex is traversed (since for each vertex, we pick some
  triangle edge to go into $B'$)
\item Since $C'$ is a simple cycle, we do not visit any edge $e_v$
  twice, hence in $C$, no vertex is repeated.
\end{enumerate}

Finally, we merely note that since \thenetg{} is subcubic and planar,
the constructed graph $G'$ is a subcubic planar graph (the connection
between two \thenetgs{} is by identifying two pendant vertices,
resulting in a single degree-2 vertex).
\end{proof}


%
%
%
%

\subsection{Generalized Location Constrained Shortest Path is \classNP-complete}

We show that \ProblemName{Generalized Location Constrained Shortest Path} is \classNP-complete.



\begin{proposition}
  \label{prop:glcsp-np-c}
  \ProblemName{Generalized Location Constrained Shortest Path} is
  \classNP-complete.
\end{proposition}
\begin{proof}
  The decision version of the problem is clearly in \classNP.
  For showing \classNP-hardness, we reduce from \ProblemName{Directed Planar Hamiltonian $s$-$t$-Path}.
  Let $(D=(V,A), s,t)$ be the input instance.

  Our reduction simply replaces every vertex $v \in V \setminus \{s,t\}$ with the following gadget.
  See \Cref{fig:glcsp-reduction} for an illustration.
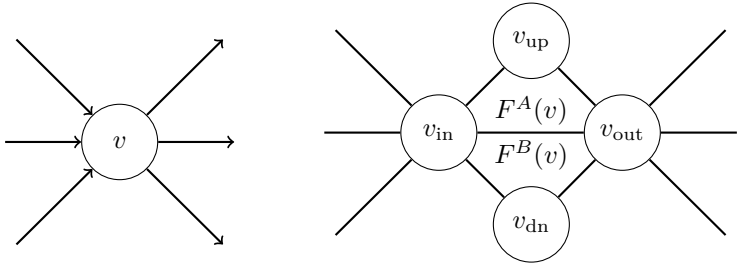
\begin{figure}
    \centering
\begin{tikzpicture}
  \tikzset{main node/.style={circle,draw,minimum size=1cm,inner sep=0pt}}
  \node[main node] (v) {$v$};
  \node[] (l1) [below left = 1cm and 1cm of v] {};
  \node[] (l2) [above left = 1cm and 1cm of v] {};
  \node[] (l3) [left = 1cm of v] {};

  \node[] (r1) [below right = 1cm and 1cm of v] {};
  \node[] (r2) [above right = 1cm and 1cm of v] {};
  \node[] (r3) [right = 1cm of v] {};

  \path[draw,thick,->] (l1) edge node {} (v);
  \path[draw,thick,->] (l2) edge node {} (v);
  \path[draw,thick,->] (l3) edge node {} (v);
  \path[draw,thick,->] (v) edge node {} (r1);
  \path[draw,thick,->] (v) edge node {} (r2);
  \path[draw,thick,->] (v) edge node {} (r3);
\end{tikzpicture}%
\qquad
\begin{tikzpicture}
  \tikzset{main node/.style={circle,draw,minimum size=1cm,inner sep=0pt}}
  \node[main node] (vi) {$\subin{v}$};
  \node[main node] (vu) [above right = 0.5cm and 0.5cm of vi]{$v_\text{up}$};
  \node[main node] (vd) [below right = 0.5cm and 0.5cm of vi]{$v_\text{dn}$};
  \node[main node] (vo) [right = 1.4cm of vi]{$\subout{v}$};
  \node[] (l1) [below left = 1cm and 1cm of vi] {};
  \node[] (l2) [above left = 1cm and 1cm of vi] {};
  \node[] (l3) [left = 1cm of vi] {};

  \node[] (r1) [below right = 1cm and 1cm of vo] {};
  \node[] (r2) [above right = 1cm and 1cm of vo] {};
  \node[] (r3) [right = 1cm of vo] {};

  \node[] [below = 0.1cm of vu] {$F^A(v)$};
  \node[] [above = 0.1cm of vd] {$F^B(v)$};

  \path[draw,thick] (vi) edge node {} (vu);
  \path[draw,thick] (vi) edge node {} (vd);
  \path[draw,thick] (vi) edge node {} (vo);
  \path[draw,thick] (vo) edge node {} (vu);
  \path[draw,thick] (vo) edge node {} (vd);

  \path[draw,thick] (l1) edge node {} (vi);
  \path[draw,thick] (l2) edge node {} (vi);
  \path[draw,thick] (l3) edge node {} (vi);
  \path[draw,thick] (vo) edge node {} (r1);
  \path[draw,thick] (vo) edge node {} (r2);
  \path[draw,thick] (vo) edge node {} (r3);
\end{tikzpicture}
\caption{Reduction from \ProblemName{Directed Planar Hamiltonian $s$-$t$-Path} to \ProblemName{Generalized Location Constrained Shortest Path}.}
\label{fig:glcsp-reduction}
\end{figure}%
  The gadget for $v$ consists of four vertices $\subin{v}$,
  $\subout{v}$, $v_\text{up}$, and $v_\text{dn}$ and the set of edges
  $$E_G(v) = \{
  \{\subin{v}, \subout{v}\},
  \{\subin{v}, v_\text{up}\},
  \{\subin{v}, v_\text{dn}\},
  \{v_\text{up}, \subout{v}\},
  \{v_\text{dn}, \subout{v}\}
  \}.$$
  We denote by~$F^A(v)$ the new face incident to~$v_\text{up}$ and by~$F^B(v)$ the new face incident to~$v_\text{dn}$.
  We embedd all gadgets in the plane in such a way that~$\{\subin{v},\subout{v}\}$ is a horizontal line with~$\subin{v}$ on the left side and~$F^A(v)$ above~$\{\subin{v},\subout{v}\}$.
  To finish the graph of the constructed instance of \ProblemName{Generalized Location Constrained Shortest Path}, we replace each arc~$(u,v)$ in the original graph with the undirected edge~$\{\subout{u},\subin{v}\}$.
  Finally, we complete the construction by setting~$\FA = \bigcup_{v \in V} F^A(v)$ and~$\FB = \bigcup_{v \in V} F^B(v)$.
  The vertices~$s$ and~$t$ remain the same in both instances and we assume without loss of generality that~$s$ has no incomming arcs and~$t$ has no outgoing arcs in the original instance.
 

	Since the reduction can clearly be computed in polynomial time, it only remains to show that $(D, s, t)$ is a yes-instance for \ProblemName{Directed Planar Hamiltonian $s$-$t$-Path} if and only if~$(G,s,t,\FA,\FB)$ is a yes-instance for \ProblemName{Generalized Location Constrained Shortest Path}.
	To this end, note that for any edge~$e$ with incident faces~$F_1$ and~$F_2$, if~$F_1 \in \FA$ and~$F_2 \in \FB$, then any solution has to traverse $e$ such that~$F_1$ on the left-hand side.

$\Rightarrow:$
Let $P$ be a directed Hamiltonian $s$-$t$-path in $D$.
We construct $P'$ by replacing every vertex $v$ in~$P$ (except for~$s$ and~$t$) with vertices~$\subin{v}$ and~$\subout{v}$ and with the edge~$\{\subin{v},\subout{v}\}$ and replace each arc~$(u,v)$ in~$P$ with the edge~$\{\subout{u},\subin{v}\}$.
Since $P$ is a simple path, so is $P'$.
Moreover, by construction of~$\FA$ and~$\FB$, $F^A(v)$ will always be on the left-hand side of $P'$ and $F^B(v)$ will always be on the right-hand side of~$P'$.

$\Leftarrow:$
Let~$P'$ be a shortest interior~$s$-$t$-path in the constructed graph.
Since~$P'$ separates each~$F^A(v)$ from its corresponding~$F^B(v)$, the edge~$\{\subin{v},\subout{v}\}$ must be contained in~$P'$.
Since~$s$ has no incomming arcs, the first edge in~$P'$ must be between~$s$ and some vertex~$\subin{v}$.
Since the edge~$\{\subin{v},\subout{v}\}$ is contained in~$P'$ and since~$P'$ is a simple path, this edge must be the second edge in~$P'$.
It is easy to prove via induction that the edges in~$P'$ now alternate between~$\{\subout{v},\subin{u}\}$ and~$\{\subin{u},\subout{u}\}$ for some vertices~$u,v \in V$.
Hence, we can construct a path from~$P'$ by replacing the two edges and the three incident vertices by the arc~$(v,u)$ and the vertex~$u$.
Since the path~$P'$ contains each edge~$\{\subin{v},\subout{v}\}$, starts in~$s$, ends in~$t$, and is a simple path, it follows that the constructed path is a directed Hamiltonian $s$-$t$-path in $D$.
\end{proof}


\end{document}